\documentclass[11pt]{article}

\usepackage[margin=1in]{geometry}
\usepackage{bigints}
\usepackage{tikz}
\usepackage{bbm}
\usetikzlibrary{decorations.pathreplacing,angles,quotes}
\usetikzlibrary{arrows}
\usepackage{comment}
\usepackage{multirow}
\usepackage{enumitem}
\usepackage{float}
\usepackage{bm}
\usepackage{multirow}
\usepackage{graphicx}
\usepackage[noline,boxed,linesnumbered, noend]{algorithm2e}
\usepackage{latexsym}
\usepackage{amsmath,amsfonts,amssymb}
\usepackage{verbatim}
\usepackage{amsthm}
\usepackage[framemethod=tikz]{mdframed}

\makeatletter
\newtheorem*{rep@theorem}{\rep@title}
\newcommand{\newreptheorem}[2]{%
\newenvironment{rep#1}[1]{%
 \def\rep@title{#2 \ref{##1}}%
 \begin{rep@theorem}}%
 {\end{rep@theorem}}}
\makeatother

\SetKwInOut{Initialization}{Initialization}

\usetikzlibrary{patterns}

\newtheorem{definition}{Definition}
\newtheorem{theorem}{Theorem}
\newtheorem{lemma}[theorem]{Lemma}
\newtheorem{remark}{Remark}
\newtheorem{observation}{Observation}

\newtheorem{proposition}{Proposition}

\newreptheorem{proposition}{Proposition}
\allowdisplaybreaks[4]

\newreptheorem{lemma}{Lemma}

\usepackage[framemethod=tikz]{mdframed}
\newmdtheoremenv[backgroundcolor=red!10,outerlinecolor=black,innertopmargin = \topskip,splittopskip = \topskip,ntheorem = true,skipabove = \baselineskip,skipbelow = \baselineskip,roundcorner=4]{Que}{Question}

\newcommand{\eat}[1]{}

\newcommand{\optim}{{\rm OPT}}

\let\oldnl\nl
\newcommand{\nonl}{\renewcommand{\nl}{\let\nl\oldnl}}

\newcommand{\algalg}{{\rm ALG}}
\newcommand{\algsol}{{\rm SOL}}
\newcommand{\algsla}{{\rm SLA}}
\newcommand{\algma}{{\rm MA}}
\newcommand{\algbpsla}{{\rm BPSLA}}
\newcommand{\algbpma}{{\rm BPMA}}

\author{Yossi Azar\thanks{Tel Aviv University. Email: azar@tauex.tau.ac.il. Supported in part by the Israel Science Foundation (grant No. 2304/20 and grant No. 1506/16).}
\and Runtian Ren\thanks{Tel Aviv University. Email: runtianren@mail.tau.ac.il.}
\and Danny Vainstein\thanks{Tel Aviv University. Email: dannyvainstein@gmail.com}}

\date{}
\title{The Min-Cost Matching with Concave Delays Problem}

\begin{document}
\maketitle
\begin{abstract}
We consider the problem of online min-cost perfect matching with concave delays. We begin with the single location variant. Specifically, requests arrive in an online fashion at a single location. The algorithm must then choose between matching a pair of requests or delaying them to be matched later on. The cost is defined by a concave function on the delay. Given linear or even convex delay functions, matching any two available requests is trivially optimal. However, this does not extend to concave delays. We solve this by providing an $O(1)$-competitive algorithm that is defined through a series of delay counters.

Thereafter we consider the problem given an underlying $n$-points metric. The cost of a matching is then defined as the connection cost (as defined by the metric) plus the delay cost. Given linear delays, this problem was introduced by Emek et al. and dubbed the Min-cost perfect matching with linear delays (MPMD) problem. Liu et al. considered convex delays and subsequently asked whether there exists a solution with small competitive ratio given concave delays. We show this to be true by extending our single location algorithm and proving $O(\log n)$ competitiveness. Finally, we turn our focus to the bichromatic case, wherein requests have polarities and only opposite polarities may be matched. We show how to alter our former algorithms to again achieve $O(1)$ and $O(\log n)$ competitiveness for the single location and for the metric case.
\end{abstract}

\section{Introduction}
In recent years, many well-known commercial platforms for matching customers and suppliers have emerged, an example of such is the ride-sharing platform - Uber.  
The suppliers first register on the platforms for the purpose of providing services to customers. The platforms are then in charge of assigning customers to the suppliers. In this way, the customers, the suppliers and the platforms may all benefit. Thanks to the rapid internet development, many of these platforms match suppliers and customers in a real time fashion. Typically, once a customer demands a request of service, the platform will assign an available supplier to this customer immediately. However, this may not always be the case - it may happen that at the time of the customer's request, there is a shortage of suppliers that meet the customer's demands (e.g., in the case of Uber, this may happen during rush hour). This usually results in a delayed response from the platform. Unfortunately, as the delay increases, so does the customer's frustration. Analogously, a supplier may also become frustrated if they are forced to wait for an assignment from the platform. 
This all boils down to the following issue faced by the platform: how can the platform match customers with suppliers in a way that the total delay (both with respect to the customer and with respect to the supplier) is minimized?

A similar issue is also faced by the online gaming platforms, which provide matching services among players to form gaming sessions. On such platforms, players arrive one by one with the intention of participating in a gaming session with other players. Typically, a gaming session consists of two players (e.g., chess, go, etc). The platform is thus required to match the players into pairs. Generally, players are willing to accept a delay of some sort. However, this too, has its limits. This results in the same issue as before: how can the gaming platforms match players in a way that the total delay time is minimized?

In this paper, we define the following problems of single location online min-cost perfect matching with delays to model the above issues. Formally, the input is a set of requests arriving in an online fashion at a single location. In the monochromatic setting, any two requests can be matched into a pair (as in the gaming platforms); in the bichromatic setting, each request has a polarity (either positive or negative) and only requests of different polarities may be matched (as in supplier-customer platforms). At each point in time, the algorithm must choose between matching a pair of requests or delaying them to be matched at a later time. Let $D(t): t \to \mathbb{R}$ denote the delay function. The cost incurred by a matching is defined as the sum of delay times: $\sum_r D(m(r) - a(r))$ where $a(r)$ and $m(r)$ denote the request's arrival and matching times. The goal is then to minimize this value.


Clearly, if $D(\cdot)$ is convex (and in particularly linear), simply matching the requests greedily (i.e., delay only if you must) yields an optimal matching in both settings. However, if $D(\cdot)$ is concave, this algorithm leads to unbounded competitiveness.\footnote{Note that concave delays may indeed be witness in real life scenarios. For example, consider a ride sharing platform. When  passengers request a ride, they shall be eager get assigned a car immediately. Thus, their frustration may increase rather fast at the beginning. However, if they are not assigned a car for a while, they may be inclined to further wait without a large increase in their frustration (since they have already waited for a long time).} To see this, simply consider the case where $2m$ requests arrive at times $\{0\} \cup \{i, i+\varepsilon\}_{i = 1}^{2m-1} \cup \{2m\}$ and $D(t) = t$ when $t \leq 1$ and $D(t) = 1$ when $t > 1$ (in the bichromatic setting, all the requests with odd indexes are positive and the other requests are negative). While clearly the greedy matching results in a cost of $m$, the optimal solution may match the first request to the last and the rest greedily - resulting in a cost of $1 + (m-1)\varepsilon$. 
Thus, the greedy algorithm's competitive ratio is at least $m$ (by setting $\varepsilon$ small enough), which can be arbitrarily large. This suggests that in order to overcome this lower bound, one must sacrifice present losses in order to benefit in the future.

Next, we generalize the above problems by considering the metric case, where the requests may appear at different locations on an underlying $n$-point metric. If two requests located at different places are matched into a pair, then in addition to their delay costs, a connection cost (as defined by the metric) is also incurred. The goal is then to minimize the total delay cost plus the total connection cost of the matching. We denote these problems as Concave MPMD and Concave MBPMD (in the monochromatic/bichromatic settings respectively). The single location variants we denote simply as the Single Location Concave MPMD (or MBPMD).


In fact, the Concave MBPMD problem can be better used to capture realistic issues faced by ride-sharing platforms: once a customer and supplier (driver) are matched, the drive must first pick the customer up, thereby incurring a cost relative to the distance between the two. This is captured by the connection cost introduced by the metric. The Concave MPMD problem may also help to better capture realistic scenarios. Consider the game of Chess as an example. In Chess, every player has a rating (based on previous matches, among other features) and players usually tend to prefer playing matching against similarly rated players. Therefore, it is a goal of the platform to match players while taking their different features into account (which may be captured through a metric) while minimizing the total delay time.

The MPMD problem was first defined by Emek et al. \cite{Online_matching_haste_makes_waste} who considered the case where $D(\cdot)$ is linear. Later, Liu et al. \cite{Impatient_Online_Matching} considered the problem of Convex MPMD where $D(\cdot)$ is convex and further posed the question of whether there exists a solution with a small competitive ratio for the concave case.
In this paper, we provide an affirmative answer to this question by providing such algorithms for both the Concave MPMD and the Concave MBPMD problems.\\

\noindent \textbf{Our Contributions:} In this paper we provide the following results.
\begin{itemize}
    \item For both the monochromatic and bichromatic single-location problems, we present counter-based $O(1)$-competitive deterministic algorithms (one for each case).
    \item For both the monochromatic and bichromatic metric problems, we present an $O(\log n)$-competitive randomized algorithms (one for each case) that generalize our Single Location algorithm in order to handle the connection costs.
\end{itemize}

\noindent \textbf{Our Techniques:} As a first step in tackling any of the formerly discussed problems, we first reduce our delay function $D(\cdot)$ to a piece-wise linear function with exponentially decreasing slopes, $f(\cdot)$. Thereafter we consider each of the problems separately but incrementally (using the techniques introduced in the single-location sections towards the metric sections).
\begin{itemize}
    \item \textbf{Monochromatic}: For this problem we introduce an $O(1)$-competitive algorithm. Recall that the optimal strategy for linear delays is simply to match any two available requests immediately. Unfortunately, this fails when considering concave delays. However, inspired by this observation, the key idea of our algorithm is to categorize our requests (based on the linear pieces of $f(\cdot)$) and then match any two available requests of the same category immediately. In order to do so, we utilize a set of counters, each counter corresponding to a linear piece of $f(\cdot)$, and we let the requests traverse through these counters, delaying them at each counter for a preset time. We then analyze the algorithm's performance by charging its increase in counters to the delay incurred by each request by the optimal solution (separately for requests that incur a large or small delay compared to the size of the counter). \\
    We then introduce an underlying metric. As a first step we embed our metric into HSTs (formally defined later). Now, in contrast to many prior algorithms that directly solve the problem with respect to the metric defined by the HST, we create a novel object that mixes the metric and delay in the following way. We first create edge counters defined by the edges in the HST. We then create delay counters, as defined in the single location. Finally, we combine the two sets of counters and define a new rooted tree of counters. We then use this rooted tree object in order to dictate the way in which requests should traverse the counters. We note that although the resulting tree may be unbounded with respect to the size of the original metric, leveraging the techniques from the single location case, we show that the algorithm is in fact bounded with respect to $h$ - the height of the original HST. We then apply the techniques as seen in 
    \cite{Online_matching_haste_makes_waste, Polylogarithmic_Bounds_on_the_Competitiveness_of_Min_cost_Perfect_Matching_with_Delays, Min_Cost_Bipartite_Perfect_Matching_with_Delays} in order to prove that our algorithm is in fact $O(\log n)$-competitive with respect to general metrics.
    
    
    \item \textbf{Bichromatic}: For this problem we introduce an $O(1)$-competitive algorithm as well. Again, we define a set of counters and let the requests traverse through them. However, due to the constraint that only requests of opposite polarities may be matched, the algorithm splits the counters in this case to positive and negative counters. In contrast to the monochromatic case, the restriction here may cause a build up of same-polarity requests on a single counter. As it turns out, the ideal strategy in this case is to only let one request at a time, traverse to the next counter. Finally, the algorithm only matches requests if they appear on the corresponding opposite-polarity counters. To analyze the algorithm's performance we use the weighted discrepancy potential function (e.g., see \cite{Min_Cost_Bipartite_Perfect_Matching_with_Delays}) to show a guarantee on the momentary delay cost of the algorithm. Thereafter, we charge the different types of counter increase to the delay incurred by the different requests in the optimal solution.\\
    Next, we introduce an underlying metric. Again, we first embed our metric into HSTs. As in the monochromatic case, we make use of the same rooted tree of counters (both edge and delay) in order to guide our requests through the different counters. Once again, we show that the algorithm's competitive ratio is bounded with respect to the original HST metric. Finally, we follow the same steps as in the monochromatic case to convert this bound to a bound for general metrics, resulting in an $O(\log n)$-competitive algorithm for the bichromatic metric case.

\end{itemize}

We note that these problems seem to be individually relevant, in that neither the monochromatic metric problem is a generalization nor a special case of the bichromatic metric problem. However, both cases may be solved by using the techniques used in the corresponding single location problems. In the single location case, however, the monochromatic problem can be directly inferred from the bichromatic problem.

\noindent \textbf{Related Work:} The problem of MPMD with linear delays was first introduced by Emek et al. \cite{Online_matching_haste_makes_waste}. In their paper they presented a randomized algorithm that achieves a competitive ratio of $O(\log^2 n + \log \Delta)$, where $\Delta$ is the ratio between the maximum and minimum distances in the underlying metric. Later, Azar et al. \cite{Polylogarithmic_Bounds_on_the_Competitiveness_of_Min_cost_Perfect_Matching_with_Delays} improved the competitive ratio to $O(\log n)$ thereby removing the dependence of $\Delta$ in the competitive ratio. They further presented a lower bound of $\Omega(\sqrt{\log n})$ on the competitiveness of any randomized algorithm. 
This lower bound was later improved to $\Omega(\log n / \log \log n)$ by Ashlagi et al. \cite{Min_Cost_Bipartite_Perfect_Matching_with_Delays} thereby nearly closing the gap.

In later work, Liu et al.\cite{Impatient_Online_Matching} considered the MPMD problem with convex delays. They first showed a lower bound of $\Omega(n)$ on the competitive ratio of any deterministic algorithm. Specifically, the lower bound constituted of an $n$-point uniform metric and a delay function of the form $D(t) = t^{\alpha}$ for $\alpha > 1$. They then went on to present a deterministic algorithm that achieves a competitive ratio of $O(n)$ for any uniform metric space and any delay function of the form $D(t) = x^\alpha$. 

Another related line of work considers MBPMD, the bipartite version of MPMD with linear delays, where requests may only be matched if they are of opposite polarity. In this case, Ashlagi et al.\cite{Min_Cost_Bipartite_Perfect_Matching_with_Delays} presented a lower bound of $\Omega(\sqrt{\log n / \log \log n})$ on any randomized algorithm. They further presented two algorithms achieving a competitive ratio of $O(\log n)$ - the first is an adaptation of Emek et al.'s \cite{Online_matching_haste_makes_waste} algorithm to the bipartite case and the second is an adaptation of Azar et al.'s \cite{Polylogarithmic_Bounds_on_the_Competitiveness_of_Min_cost_Perfect_Matching_with_Delays} algorithm. 

The deterministic variants of MPMD and MBPMD with linear delays have also managed to spark great research interest. For further reading see \cite{A_Match_in_Time_Saves_Nine:_Deterministic_Online_Matching_with_Delays, A_Primal_Dual_Online_Deterministic_Algorithm_for_Matching_with_Delays, AzarF2018, EmekSW2019}. 

Finally, we note that although we consider the problem of matching with delays, many new online problems were also considered using the notion of delays. These problems have extensive applications in fields that range from operations management, to operating systems and supply chain management. Examples of such problems include the online services with delays problem \cite{AzarGGP2017, BienkowskiKS2018, azar2019general}, the multi-level aggregation problem  \cite{bienkowski2016online, buchbinder2017depth, azar2019general} and many more. For a more extensive list of such problems, see \cite{carrasco2018online, gupta2020caching, On_bin_packing_with_clustering_and_bin_packing_with_delays, The_Price_of_Clustering_in_Bin-Packing_with_Applications_to_Bin_Packing_with_Delays, azar2020setcoverd, azar2020beyond}.

\noindent \textbf{Paper Organization:} We start by introducing some useful notations and preliminaries in Section \ref{section-notation}. 
Then, we consider the problem in the monochromatic setting. 
We first present an $O(1)$-competitive deterministic algorithm for the single location case in Section \ref{section.single_location} and then generalize our idea to design an $O(\log n)$-competitive algorithm for the metric case in Section \ref{section.metric}. 
After that, we consider the problem in the bichromatic setting.
Similarly, in Sections \ref{section.bipartite.single_location} and \ref{section.bipartite.metric}, we present an $O(1)$-competitive deterministic algorithm and an $O(\log n)$-competitive algorithm in the single location case and the metric case respectively. 
We conclude by stating a few remarks and related open problems in Section \ref{section.conclusion}. 
\section{Notations and Preliminaries}
\label{section-notation}
We define our delay function $D(\cdot)$, such that $D(0) = 0$ and it is concave (i.e., $\forall x_1, x_2 \geq 0 \text{ and } \forall \alpha \in [0,1]: D((1-\alpha)x_1 + \alpha x_2) \geq (1-\alpha)D(x_1) + \alpha D(x_2)$). Furthermore we assume that $D(\cdot)$ is both continuous and monotonically increasing. Note that this guarantees that the one sided derivatives exist at any point. We further assume that all derivatives are positive and bounded. Finally, to ensure that the optimal matching does not keep requests unmatched indefinitely, we assume that $\lim_{t \rightarrow \infty} D(t) = \infty$.


We introduce several notations that will aid us throughout the paper. A function $f(\cdot)$ is said to be piece-wise linear if the x-axis may be partitioned such that in each interval $f(\cdot)$ is linear. The following lemma states that in fact it is enough to consider delay functions that are piece-wise linear with exponentially decreasing slopes (i.e., $\alpha_i \geq 2 \alpha_{i+1}$ where $\alpha_i$ and $\alpha_{i+1}$ denote the slopes of two consecutive linear pieces). 

\begin{lemma}
\label{lemma.approximate_concave_with_piece_wise_linear}
There exists a piece-wise linear function $f(\cdot)$ with exponentially decreasing slopes such that for any $x \geq 0$ we have $f(x) \leq D(x) \leq 2f(x)$.
\end{lemma}

Therefore, we will assume henceforth that our delay function is piece-wise linear (while incurring a multiplicative loss of 2). Formally, we define the function such that for any $i=0,1,2, \ldots,$ our function is linear during the interval $(x_i,x_{i+1})$ with slope $\alpha_i$. By Lemma \ref{lemma.approximate_concave_with_piece_wise_linear} we may assume $\alpha_i \geq 2 \alpha_{i+1}$ for all $i \in \mathbbm{N}$. We further denote $\ell_i = x_{i+1} - x_i$ and $y_i = f(x_{i+1}) - f(x_i)$. See Figure \ref{figure.piecewisedelay} for a pictorial example.

\begin{figure}[h]
\centering
\includegraphics[width= 7cm]{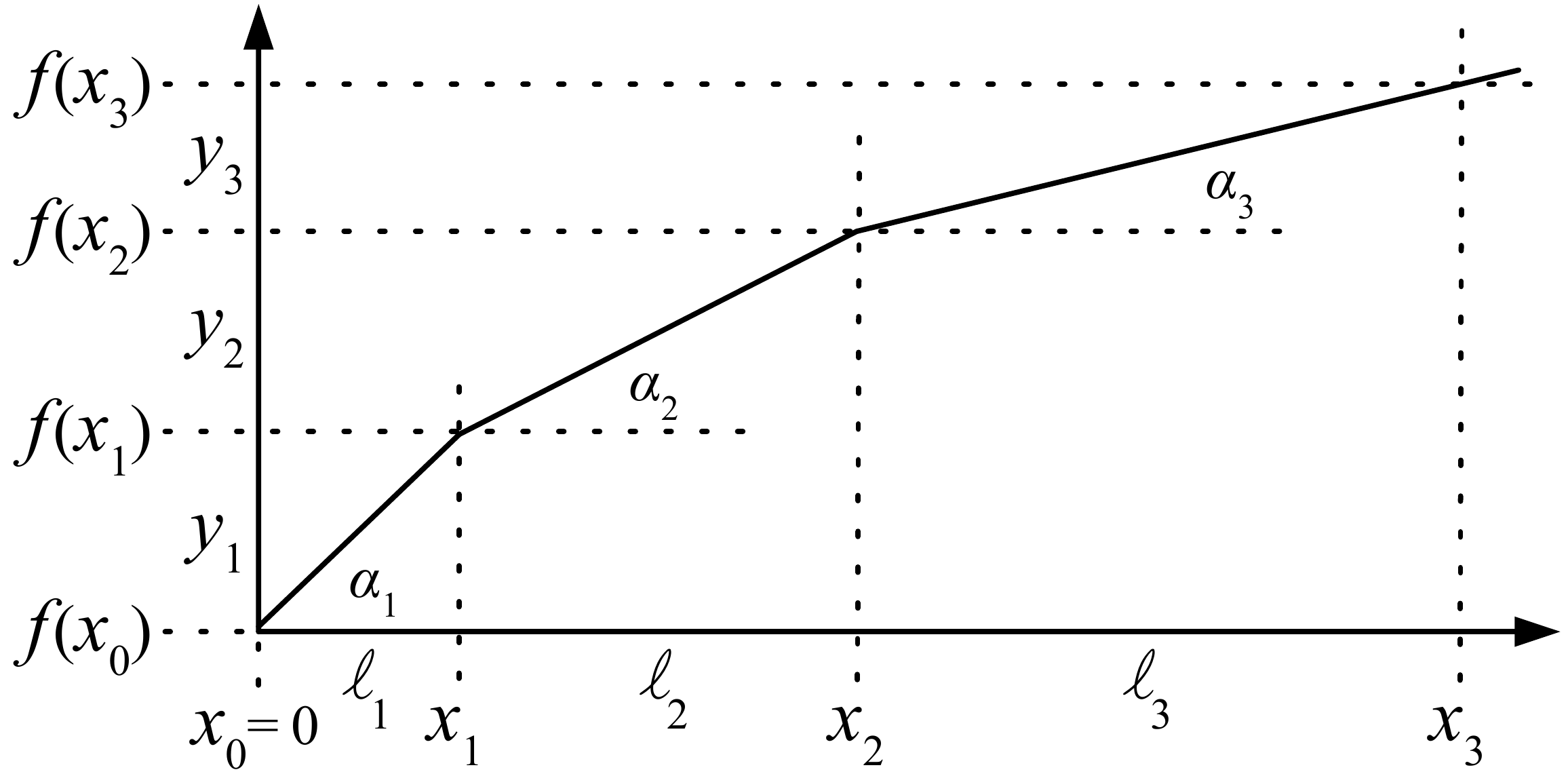}
\caption{A piece-wise linear delay function.}
\label{figure.piecewisedelay}
\end{figure}

We also make use of Hierarchical Separated Trees (HST's) in our paper. 

\begin{definition}
A weighted $\sigma$-HST is a tree metric $T = (V_T, E_T, w_T)$, defined with respect to some general metric $G=(V_G,E_G,w_G)$. The leaves of $T$ are $V_G$. Furthermore, $w_T$ defines a $\sigma$-hierarchical separation - formally, the weight between any node $v$ and any of its children is at most $1 / \sigma$ times the weight of the edge between $v$ and its parent. The distance between any two points in $V_G$ is then defined as the distance between these two points in $T$.
\end{definition}

Throughout our paper we denote the optimal matching as $\optim$. Furthermore, given an algorithm, we will abuse notation and denote by $\algalg$ both the matching produced by the algorithm and its cost.
\section{Matching on a Single Location}
\label{section.single_location}

In this section we consider the Single Location Concave MPMD problem. We begin with several notions that will be used to define our algorithms. Our algorithm makes use of counters to decide whether to delay or match requests. Specifically, we define a counter for every linear section in the delay function. We will refer to these counters as $z_1, z_2, \ldots$ Our algorithm continually moves requests across these counters; as such, we will say that ``a request is associated with a counter" and that ``a counter contains a request" at a specific time if the 
request is associated with that counter at that moment. We are now ready to define our algorithm.

First, we formally define the counter $z_k$ such that it has a capacity of $y_k$ (as defined by the delay function) and a slope of $\alpha_k$. Now, once a request arrives we add the request to $z_1$. Next, if at any point in time a counter contains 2 requests, match them and remove them from the counters. Furthermore, at any point in time, we consider all counters $z_k$ simultaneously and \textit{increase counter $z_k$ at a rate equal to $z_k$'s slope ($\alpha_k$) if and only if (1) there is a request that is associated with $z_k$ and (2) the number of requests associated with the counters $z_1, \ldots, z_{k-1}$ is even}. Finally, if at any moment any counter $z_k$ reaches its capacity, move its associated request to the next counter and reset $z_k$ to 0. Note that the moved request might move to a counter, $z_{k+1}$ that has been partially filled - in such a case the new request will continue to fill $z_{k+1}$ from that point. (This is indeed necessary since otherwise the algorithm fails on the "bad" example given in the introduction).

We note that each counter may have at most a single pending request associated with it and therefore a request may be pending if and only if it is associated with some counter. We denote this algorithm as the Single-Location-Algorithm ($\algsla$). See Algorithm \ref{algorithm.single_location} in the Appendix for a formal definition.

\begin{remark}
We remark that $\algsla$ will indeed match all requests since we remove requests from the counters only if they are paired and therefore if the counters contain a single request, another request must arrive in the future. Furthermore, the request associated with the lowest counter will always increase while the request associated with the second lowest counter will not increase. Therefore, the last 2 remaining requests will always match.
\end{remark}

\noindent We first state the main theorem of this section.

\begin{theorem}
\label{theorem.single_location_O(1)_competitive}
$\algsla$ is $O(1)$-competitive for any concave delay function.
\end{theorem}

Throughout this section we abuse notation and denote by $\algsla$ both the matching produced by the algorithm and its cost. We do the same with respect to $\optim$, the optimal matching. In order to prove Theorem \ref{theorem.single_location_O(1)_competitive} we use two steps. We first upper bound our algorithm's cost by the actual increase in its counters throughout the input. We then lower bound $\optim$'s cost by the increase in the appropriate counters. To formally define the increase in counters we introduce the following definition.

\begin{definition}
For a counter $z_k$ define $z_k'(t)$ to be $\alpha_k$ if $z_k$ increases at time $t$ and 0 otherwise. Furthermore, for a request $r$ define $z_r'(t)$ to be $\alpha_k$ if $r$ is associated with $z_k$ at time $t$ and 0 if the request is not associated with any counters.
\end{definition}

\noindent Therefore, the total increase in counters is in fact $\sum_k \int_t z_k'(t)dt$.

\begin{remark}
We note that $z_k$ is zeroed whenever a request moves up from $z_k$ to $z_{k+1}$. Therefore, $\int_t z_k'(t)dt$ might be much larger than the counter's level at a single point in time.
\end{remark}

\subsection{Upper Bounding $\algsla$'s Cost}

In this subsection we would like to upper bound $\algsla$ by the overall increase in its counters. In order to do so we first introduce the following notations that will add us in our proof. Recall that $[x_{i-1}, x_i)$ denotes the $i$'th linear piece of the delay function and that given a request $r$, $a(r)$ denotes its arrival time.

\begin{definition}
\label{definition.single_location.delta_t(r)_k_t(r)}
Given a request $r$ and time $t$ when $r$ is unmatched by $\algsla$ at this moment, we define
\begin{itemize}
    \item $\delta_t(r)$: $r$'s delay level at time $t$, i.e., $t - a(r) \in [x_{\delta-1}, x_\delta)$.
    \item $k_t(r)$: the delay counter $r$ is associated with at time $t$.
\end{itemize}
\end{definition}

\begin{definition}
\label{definition.single_location.d'_r(t)_s_r(t)}
Given a request $r$ and time $t$, we define
\begin{itemize}
    \item $d'_r(t)$: the momentary delay incurred by $r$ at time $t$ with respect to $\algsla$, i.e., $d'_r(t) = \alpha_{\delta_t(r)}$ if $r$ is unmatched at time $t$ and $d'_r(t) = 0$ otherwise.
    \item $s_r(t)$: the slope of the counter that $r$ is associated with at time $t$, i.e., $s_r(t) = \alpha_{k_t(r)}$ if $r$ is unmatched at time $t$ and $s_r(t) = 0$ otherwise.
\end{itemize}
\end{definition}

Note that Definition \ref{definition.single_location.d'_r(t)_s_r(t)} is well defined in the sense that a request is unmatched by $\algsla$ if and only if it is associated with some counter. Recall that $m(r)$ denotes the time in which a request $r$ was matched by $\algsla$. Further note that seemingly $d'_r(t) \leq s_r(t)$ at any moment $t$ since the delay counters sometimes "freeze" whereas the request's delay does not. However, it may be the case that when $r$ arrives its counters are already partially filled - in such a case $r$ will immediately move to a higher counter and at that moment we would have $d'_r(t) > s_r(t)$. Thus, we introduce the following lemma.

\begin{lemma}
\label{lemma.single_location.real_delay_bounded_by_aux}
$\sum_r \int_t d_r'(t)dt \leq \sum_r  \sum_{k = 1}^{k_{m(r)}(r)-1} y_k  + \sum_r  \int_t s_r(t)dt$.
\end{lemma}

\begin{proof}
We argue that in fact, for any request $r$, $\int_t d_r'(t)dt \leq \sum_{k = 1}^{k_{m(r)}(r)-1} y_k  + \int_t s_r(t)dt$.

If $k_{m(r)}(r) > \delta_{m(r)}(r)$ then clearly $\int_t d_r'(t)dt \leq \sum_{k = 1}^{k_{m(r)}(r)-1} y_k$. Otherwise, assume $k_{m(r)}(r) \leq \delta_{m(r)}(r)$. The value $\int_t d_r'(t)dt$ constitutes of delay accumulated up to and including the $k_{m(r)}(r)-1$ linear piece (in the delay function) and delay accumulated during the rest of the time. Hence, 
\[
\int_t d_r'(t)dt = 
\sum_{k = 1}^{k_{m(r)}(r)-1} y_k + \int_{x_{k_{m(r)}(r)}}^{m(r) - a(r)} d_r'(t)dt,
\]
where $x_{k_{m(r)}(r)}$ denotes the time that the $k_{m(r)}(r)$'th linear piece begins in the delay function (see Figure \ref{figure.piecewisedelay}).

The value $\int_{x_{k_{m(r)}(r)}}^{m(r) - a(r)} d_r'(t)dt$ is upper bounded moment-wise by $\int_t s_r(t)dt$, since the delay accumulated by $s_r(t)dt$ is at least $\alpha_{k_{m(r)}(r)}$ and the delay accumulated by $d_r'(t)dt$ during that time is at most $\alpha_{k_{m(r)}(r)}$. Thus, $\int_t d_r'(t)dt \leq \sum_{k = 1}^{k_{m(r)}(r)-1} y_k  + \int_t s_r(t)dt$. Summing over all requests,
\begin{align*}
\sum_r \int_t d_r'(t)dt \leq \sum_r  \sum_{k = 1}^{k_{m(r)}(r)-1} y_k  + \sum_r  \int_t s_r(t)dt.    
\end{align*}
\end{proof}

\begin{proposition}
\label{proposition.single_location.bound_algorithm_by_counters}
$\algsla \leq 3 \cdot \sum_k \int_t z_k'(t)dt$.
\end{proposition}

\begin{proof}
We first observe that due to the fact that once a request moves up a counter, the former counter's value is reset to 0, we have,
\begin{align}
\label{equation.single_location.depletion_is_bounded_by_real_counter_increase}
\sum_r  \sum_{k = 1}^{k_{m(r)}(r)-1} y_k \leq \sum_r \int_t z'_r(t)dt.    
\end{align}

Next, we argue that $\sum_r \int_t s_r(t)dt \leq 2 \sum_r \int_t z'_r(t)dt$. Due to the fact that $s_r(t) = 0$ if $r$ is unmatched by $\algsla$ at time $t$, by denoting $I(r)$ as all time points for which $r$ is unmatched (by the algorithm), we are guaranteed that, $\sum_r \int_t s_r(t)dt = \sum_r \int_{t \in I(r)} s_r(t)dt$.

Next, denote by $r^*_t$ the unmatched request associated with the lowest-indexed counter at time $t \in \cup_r I(r)$. Recall that there exists at most one unmatched request per counter and the counters’ slopes decrease exponentially. According to our algorithm $\algsla$, the counter containing $r^*_t$ is lowest and therefore it increased at time $t$, guaranteeing that $s_{r^*_t}(t) = z'_{r^*_t}(t)$. We thus have,
\begin{align}
\label{equation.single_location.counter_increase_bounded_by_real_counter_increase}
\sum_r \int_t s_r(t)dt &= 
\sum_r \int_{t \in I(r)} s_r(t)dt \leq
2 \int_{t \in \cup_r I(r)} s_{r^*_t}(t)dt \nonumber \\ &= 
2 \int_{t \in \cup_r I(r)} z'_{r^*_t}(t)dt \leq 
2\sum_r \int_t z'_r(t)dt.
\end{align}

\noindent Therefore, combining the above with Lemma \ref{lemma.single_location.real_delay_bounded_by_aux},
\begin{align*}
\algsla &= 
\sum_r \int_t d_r'(t)dt \leq 
\sum_r \sum_{k = 1}^{k_{m(r)}(r)-1} y_k  + \sum_r \int_t s_r(t)dt \\ &\leq
\sum_r \int_t z'_r(t)dt + \sum_r \int_t s_r(t)dt \leq
3 \sum_r \int_t z'_r(t)dt =
3 \sum_k \int_t z_k'(t)dt,
\end{align*}
where the first inequality follows from Lemma \ref{lemma.single_location.real_delay_bounded_by_aux}, the second follows from equation (\ref{equation.single_location.depletion_is_bounded_by_real_counter_increase}), the third follows from equation (\ref{equation.single_location.counter_increase_bounded_by_real_counter_increase}) and the last equality follows by rearranging the terms.
\end{proof}

\subsection{Lower Bounding $\optim$'s Cost}

In this subsection we upper bound the overall increase in counters by the optimal solution by charging the increase to different requests of the optimal solution. The following proposition states this formally.

\begin{proposition}
\label{proposition.single_location.bound_counters_by_OPT}
$\sum_k \int_t z_k'(t)dt \leq 12 \cdot \optim$.
\end{proposition}

In order to prove the proposition we consider the increase in each counter $z_k$ separately. We further split the increase in $z_k$ into phases as follows.

\begin{definition}
\label{definition.single_location.counter_intervals}
For a given counter $z_k$ let $0 < t^k_1, t^k_2, t^k_3, \ldots, t^k_{m_k-1}<\infty$ denote all points in time for which $z_k$ changes value from non-zero to zero (i.e., whenever a request moves from counter $z_k$ to $z_{k+1}$). We further denote $t^k_{m_k} = \infty$ and $t^k_0 = 0$. 
\end{definition}

Note that for counters $z_k \neq z_{k'}$ the points $t^k_i$ and $t^{k'}_i$ need not be aligned (unless $i = 0$). When it is clear from context that we are considering a specific counter $z_k$, we may denote $t^k_i$ simply by $t_i$. Before proving the proposition we first introduce several definitions and lemmas that will aid us in our proof.

Throughout the remainder of this section, given a counter $z_k$ and an interval $I_i^k$ defined with respect to $z_k$, we define $\bm{R(I_i^k)}$ to be the set of all requests that arrived during the time interval $I_i^k$.

Next, we introduce the notion of an odd-subinterval. Given an interval $I$ we would like to consider all time points $t\in I$ such that an odd number of requests from $R(I)$ have arrived. We denote the set of these points as $I_{odd}$ and refer to them as $I$'s odd-subinterval. The following definition defines this formally.

\begin{definition}
Given an interval $I^k_i = [t^k_i, t^k_{i+1})$ as defined by some counter $z_k$ we denote $|R(I_i^k)| = \gamma_k$. We further denote $R(I_i^k)$'s arrival times by $\{a_j\}_{j=1}^{\gamma_k}$. Given these notations, if $\gamma_k$ is odd then we define $I^{odd}_{i,k} = \cup_{j=1}^{\lfloor \frac{\gamma_k}{2} \rfloor}[a_{2j-1}, a_{2j}) \cup [a_{\gamma_k},t^k_{i+1})$. Otherwise, if $\gamma_k$ is even we define $I^{odd}_{i,k} = \cup_{j=1}^{\lfloor \frac{\gamma_k}{2} \rfloor}[a_{2j-1}, a_{2j})$.
\end{definition}


The following lemma (whose proof is deferred to the Appendix) states that any point for which $z_k$ increases must lie within the odd-subinterval of the corresponding interval.


\begin{lemma}
\label{lemma.single_location.interval_parity_of_SLA}
Consider some time interval $I^k_i = [t^k_i, t^k_{i+1})$ defined with respect to counter $z_k$. Let $t \in I^k_i$ denote some time for which $z_k$ increases. Therefore, $t \in I_{i,k}^{odd}$.
\end{lemma}


\begin{definition}
\label{definition.single_location.interval_with_respect_to_opt}
Given an interval $I$ and some time $t \in I$ we say that $\optim$ is live with respect to $I$ at time $t$ if $\optim$'s matching contains an unmatched request $r$ at time $t$ such that either $r \in R(I)$ or $r$'s pair (with respect to $\optim$) belongs to $R(I)$.
\end{definition}

\noindent For a pictorial example of Definition \ref{definition.single_location.interval_with_respect_to_opt} see Figure \ref{figure.single_location.definition_opt_live}.

\begin{figure}[H]
\centering
\includegraphics[width= 10cm]{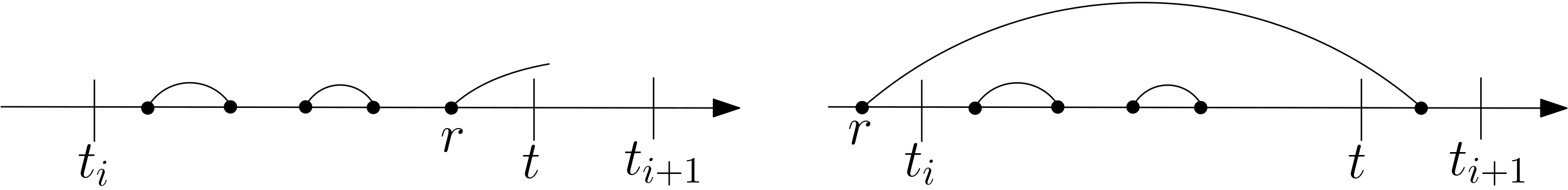}
\caption{Request $r$ causes $\optim$ to be live at time $t$ with respect to interval $I = [t_i, t_{i+1})$. In the left hand side it is since $r \in R(I)$ and the right hand side is because $r$'s pair belongs to $R(I)$.}
\label{figure.single_location.definition_opt_live}
\end{figure}

\noindent The following lemmas are combinatorial lemmas and therefore their proofs are deferred to the Appendix.

\begin{lemma}
\label{lemma.single_location.last_interval_is_even}
For any counter $z_k$, $|R(I^k_{m_k-1})|$ is even.
\end{lemma}




\begin{lemma}
\label{lemma.single_location.last_interval_parity_of_OPT}
For any counter $z_k$, if $t \in I_{m_k-1,k}^{odd}$ then $\optim$ must be live with respect to $I^k_{m_k-1}$ at time $t$.
\end{lemma}

\begin{lemma}
\label{lemma.single_location.interval_is_odd}
For any counter $z_k$ and any interval $I_i^k$, $|R(I^k_i)|$ is odd.
\end{lemma}

For counter $z_k$, consider any two consecutive intervals $I^k_i = [t^k_i, t^k_{i+1})$ and $I^k_{i+1} = [t^k_{i+1}, t^k_{i+2})$ such that $t^k_{i+2} < t^k_{m_k}$ (if such intervals exist). Let $R(I^k_i \cup I^k_{i+1})$ denote the set of requests that arrive during the interval $I^k_i \cup I^k_{i+1}$ respectively. Lemma \ref{lemma.single_location.middle_interval_parity_OPT} (whose proof is deferred to the Appendix) will aid us in charging the increase in counters towards $\optim$'s delay.


\begin{lemma}
\label{lemma.single_location.middle_interval_parity_OPT}
For $I = I^k_i$ or $I = I^k_{i+1}$ the following condition holds: \\
For any $t \in I_{odd}$, $\optim$ must be live with respect to $I_i \cup I_{i+1}$ at time $t$.
\end{lemma}

\begin{lemma}
\label{lemma.single_location.the_number_of_intervals_is_odd}
For any counter $z_k$ the number of intervals defined with respect to that counter, $m_k$, is odd.
\end{lemma}

\begin{proof}
Follows from Lemmas \ref{lemma.single_location.last_interval_is_even} and \ref{lemma.single_location.interval_is_odd}.
\end{proof}

\noindent We are now ready to prove our proposition.

\begin{proof}[Proof of Proposition \ref{proposition.single_location.bound_counters_by_OPT}]
Before charging the increase in our counters (i.e.,  $\sum_{k=1}^d \int_t z_k'(t)dt$) to $\optim$ we first need several notations. Let $\rho(r)$ denote the delay incurred by a request $r$ or the delay incurred by the request that $r$ is matched to, both with respect to  $\optim$'s matching  (i.e., for matched requests $r$ and $r'$, $\rho(r) = \rho(r')$ is defined as the delay incurred by the earlier request to arrive). Therefore,
\begin{align}
\label{equation.single_location.charging_scheme_lemma.1}
\optim = \frac{1}{2} \sum_r \rho(r).
\end{align}

We say that $r$ is of level $k$ with respect to $\optim$ if $\rho(r) \in [\sum_{j=1}^k y_j, \sum_{j=1}^{k+1} y_j)$. 

By Lemma \ref{lemma.single_location.the_number_of_intervals_is_odd} we may partition our overall set of intervals (with respect to some $z_k$) into pairs $\{I^k_{2i}, I^k_{2i+1}\}$ for $i \in \{0,1,\ldots,\frac{m_k-3}{2}\}$ with the addition of the interval $I^k_{m_k-1}$. We define our charging scheme separately for charging counter increase during $I^k_{m_k-1}$ (for each counter $z_k$) and for charging counter increase during $\{I^k_{2i}, I^k_{2i+1}\}$. For any counter $z_k$ and time $t$ in which $z_k$ increases, we charge the momentary increase to some request $r$ and denote our charging scheme by $f_k(t) = r$.

For any counter $z_k$ and any pair of intervals $\{I^k_{2i}, I^k_{2i+1}\}$, the total increase in $z_k$ is exactly $2y_k$. Therefore, it is enough to only charge one of the intervals to $\optim$ (and incur a multiplicative factor of 2 later on). We will choose the interval to charge as follows. By Lemma \ref{lemma.single_location.middle_interval_parity_OPT} one of these intervals guarantees the condition as defined in the lemma. We assume it is the first (i.e., $I^k_{2i}$) and use that interval towards our charging scheme (if it were the second we would have used that interval and continued identically). Therefore we will define $f_k(t)$ only for $t \in I^k_{2i}$ for $i \in \{0,1,\ldots,\frac{m_k-3}{2}\}$, such that $z_k$ increases at time $t$. To ease notation in the following definition, let $I^k_{m_k} = \emptyset$.

\noindent \textbf{Defining} $\mathbf{f_k(t)}$: For any $i \in \{0,1,\ldots,\frac{m_k-1}{2}\}$ and any $t \in I^k_{2i}$ such that $z_k$ increases, we define $f_k(t)$ as follows. We set $f_k(t) = r$ to $r \in R(I^k_{2i} \cup I^k_{2i+1})$ such that $r$ is of level $\geq k$ with respect to $\optim$, if such a request exists. Furthermore we break ties by taking the earliest such request to arrive. If such a request does not exist, then we will charge the increase to the request $r$ that causes $\optim$ to be live with respect to $I^k_{2i} \cup I^k_{2i+1}$, at time $t$.

\begin{remark}
$f_k(t)$ is indeed well defined, i.e., there exists a request $r$ satisfying at least one of the predefined conditions. This is true for any $t \in I^k_{2i}$ due to the fact that if $z_k$ increases then $t \in I_{2i,k}^{odd}$ (Lemma \ref{lemma.single_location.interval_parity_of_SLA}) which in turn guarantees that $\optim$ is live with respect to $I^k_{2i} \cup I^k_{2i+1}$ (due to Lemma \ref{lemma.single_location.middle_interval_parity_OPT} and the fact that we chose the interval that guarantees the defined condition). This is similarly true for any $t \in I^k_{m_k-1}$ due to Lemmas \ref{lemma.single_location.interval_parity_of_SLA} and \ref{lemma.single_location.last_interval_parity_of_OPT}.
\end{remark}

For any counter $z_k$ let $D_k$ denote the set of requests within the image of $f_k(\cdot)$. Note that $f_k(\cdot)$ is defined only for a single interval from each pair, $I^k_{2i}$ and $I^k_{2i+1}$, and the increase in $z_k$ is exactly the same in each such interval (specifically, it is $y_k$). Therefore,
\begin{align}
\label{equation.single_location.charging_scheme_lemma.2}
\sum_{k=1}^d \int_t z_k'(t)dt \leq 2 \sum_{k=1}^d \sum_{r \in D_k} \int_{t \in f^{-1}_k(r)} z_k'(t)dt.
\end{align}

For a request $r$ let $L^*(r)$ denote its level with respect to $\optim$. Therefore, we may sum over all levels of $\optim$'s requests, $k^*$ and get,

\begin{align}
\label{equation.single_location.charging_scheme_lemma.3}
\sum_{k=1}^d \sum_{r \in D_k} \int_{t \in f^{-1}_k(r)} z_k'(t)dt &= 
\sum_{k^*=1}^d \sum_{r: L^*(r)=k^*} \sum_{k \in [d]: r \in D_k \land k \leq k^*}\int_{t \in f^{-1}_k(r)} z_k'(t)dt \nonumber \\ &+ 
\sum_{k^*=1}^d \sum_{r: L^*(r)=k^*} \sum_{k \in [d]: r \in D_k \land k > k^*}\int_{t \in f^{-1}_k(r)} z_k'(t)dt.
\end{align}

Due to our partition into intervals, for a given counter $z_k$, every request $r \in D_k$ is only charged by the increase in $z_k$ from within the interval $r$ belongs to. Due to the fact that $z_k$ may increase by at most $y_k$ within each interval, we have $\forall k \forall r \in D_k: \int_{t \in f^{-1}_k(r)} z_k'(t)dt \leq y_k$. Recall that if $L^*(r) = k^*$ then $\rho(r) \in [\sum_{k=1}^{k^*} y_k, \sum_{k=1}^{k^*+1} y_k)$. Thus, overall, for any request $r$ with $L^*(r) = k^*$,

\begin{align}
\label{equation.single_location.charging_scheme_lemma.4}
\sum_{k \in [d]: r \in D_k \land k \leq k^*} \int_{t \in f^{-1}_k(r)} z_k'(t)dt \leq \sum_{k \in [d]: r \in D_k \land k \leq k^*} y_k \leq \sum_{j = 1}^{k^*} y_j \leq \rho(r).
\end{align}

For brevity let $\cup_{k \in [d]: r \in D_k \land k > k^*} f_k^{-1}(r) = \cup f_k^{-1}(r)$. For any $t \in \cup f_k^{-1}(r)$ let $k_{\ell}(t)$ denote the lowest counter of $\{k \in [d]: r \in D_k \land k > i\}$. Since $\alpha_i$ decrease exponentially,
\begin{align}
\label{equation.single_location.charging_scheme_lemma.5}
\sum_{k \in [d]: r \in D_k \land k > k^*}\int_{t \in f^{-1}_k(r)} z_k'(t)dt \leq 2 \int_{t \in \cup f^{-1}_k(r)} z_{k_{\ell}(t)}'(t)dt.
\end{align}

By the definition of $f_k(\cdot)$ we know that for every $t \in \cup f_k^{-1}(r)$, $\optim$ must have paid a momentary delay towards $\rho(r)$ (this is due to the fact that $k > k^*$  and therefore the second condition in $f_k(\cdot)$ must be satisfied). Since $k > k^*$, $z_{k_{\ell}(t)}'(t) \leq \alpha_{k^*+1}$ for all $t \in \cup f_k^{-1}(r)$. On the other hand, delaying $r$ costs at least $\alpha_{k^*+1}$ moment-wise since its level is $k^*$. Therefore,
\begin{align}
\label{equation.single_location.charging_scheme_lemma.6}
\int_{t \in \cup f^{-1}_k(r)} z_{k_{\ell}(t)}'(t)dt \leq \int_{t \in \cup f^{-1}_k(r)} \alpha_{k^*+1} \leq \rho(r).
\end{align}

\noindent Combining all of the above yields,
\begin{align*}
\sum_{k=1}^d \int_t z_k'(t)dt &\leq
2 \sum_{k^*=1}^d \sum_{r: L^*(r)=k^*} \sum_{k \in [d]: r \in D_k \land k \leq k^*}\int_{t \in f^{-1}_k(r)} z_k'(t)dt \\ &+ 
2 \sum_{k^*=1}^d \sum_{r: L^*(r)=k^*} \sum_{k \in [d]: r \in D_k \land k > k^*}\int_{t \in f^{-1}_k(r)} z_k'(t)dt \\ &\leq
2 \sum_{k^*=1}^d \sum_{r: L^*(r)=k^*} \rho(r) + 4 \sum_{k^*=1}^d \sum_{r: L^*(r)=k^*} \rho(r) \\ &=
6 \sum_r \rho(r) = 12 \cdot \optim,
\end{align*}
where the first inequality is due to (\ref{equation.single_location.charging_scheme_lemma.2}) (\ref{equation.single_location.charging_scheme_lemma.3}), the second inequality is due to (\ref{equation.single_location.charging_scheme_lemma.4})  (\ref{equation.single_location.charging_scheme_lemma.5}) (\ref{equation.single_location.charging_scheme_lemma.6}) and the last equality is due to (\ref{equation.single_location.charging_scheme_lemma.1}).
\end{proof}

\noindent We are now ready to prove our main theorem - Theorem \ref{theorem.single_location_O(1)_competitive}.

\begin{proof}[Proof of Theorem \ref{theorem.single_location_O(1)_competitive}]
Since $\sum_r \int_t z_r'(t)dt = \sum_k \int_t z_k'(t)dt$, Propositions \ref{proposition.single_location.bound_algorithm_by_counters} and \ref{proposition.single_location.bound_counters_by_OPT} yield the theorem.
\end{proof}

\section{Matching on a Metric}
\label{section.metric}

In this section we consider the Concave MPMD problem. Recall that in this problem the cost of a matching is comprised of both a connection cost and a delay cost. We solve this problem by first reducing the general metric to metrics based on weighted $\sigma$-HSTs and then define an algorithm for such a case.

\subsection{Reduction to Weighted $\sigma$-HSTs}
\label{subsection.metric.reduction_to_HSTs}

In order to reduce our problem to metrics defined by weighted $\sigma$-HSTs we follow the works of Emek et al. \cite{Online_matching_haste_makes_waste} and Azar et al. \cite{Polylogarithmic_Bounds_on_the_Competitiveness_of_Min_cost_Perfect_Matching_with_Delays} . We first introduce the definition of an $(\beta, \gamma)$-competitive algorithm. 

\begin{definition}
Let $\sigma$ denote an instance of Concave MPMD, let $\algalg$ denote a randomized online algorithm solving the problem and let $\algsol$ denote an arbitrary solution to the problem. Further denote by $\algsol_c$ and $\algsol_d$ the connection and delay costs incurred by the solution, $\algsol$. We say that $\algalg$ is $(\beta, \gamma)$-competitive if $\mathbb{E}[\algalg(\sigma)] \leq \beta \cdot \algsol_c + \gamma \cdot \algsol_d$.
\end{definition}

Next we randomly embed our general metric into weighted $\sigma$-HSTs. First, we formally define the notion of embedding a metric into another metric.

\begin{definition}
Let $G = (V_G, E_G, w_G)$ be a finite metric space and let $D$ be a distribution over metrics on $V_G$. We say that $G$ embeds into $D$ if for any $x,y \in V_G$ and any metric space $H = (V_G, E_H, w_H)$ in the support of $D$, we have $w_G(x,y) \leq w_H(x,y)$. We define the embedding's distortion as
\[
\mu = \max_{x\neq y \in V(G)} \frac{\mathbb{E}_{H \sim D}[w_H(x,y)]}{w_G(x,y)}.
\]
\end{definition}

We use the following lemma (introduced by Bansal et al. \cite{A_Polylogarithmic_Competitive_Algorithm_for_the_k_Server_Problem}) in order to embed our metric $G$ into a distribution over weighted 2-HSTs with height $O(\log n)$ where is $n = |V_G|$.

\begin{lemma}
\label{lemma.embed_metric_to_hst}
Any $n$-point metric $G$ can be embedded, with distortion $O(\log n)$, into a distribution $D$ supported on metrics induced by weighted 2-HSTs with height $O(\log n)$.
\end{lemma}

Next we use the following lemma (introduced by Emek et al. \cite{Online_matching_haste_makes_waste}) that shows how to convert a matching's cost from a metric in the support of the embedding to a cost on the original metric. Note that Emek et al. proved this for linear delay functions, however their result holds for concave delay functions as well. The proof is identical and is therefore omitted.

\begin{lemma}
\label{lemma.transfer_competitiveness_from_hst_to_general_metric}
Suppose that a metric $G$ can be embedded into a distribution $D$ supported on metric spaces over $V_G$ with distortion $\mu$. Additionally suppose that for every metric space in the support of $D$ exists a deterministic algorithm that is $(\beta, \gamma)$-competitive for the Concave MPMD problem. Therefore, there exists a $(\mu \beta, \gamma)$-competitive algorithm for the Concave MPMD problem on the metric $G$.
\end{lemma}

Therefore, it is enough to show that an algorithm is $(O(1), O(h))$-competitive on metrics defined by weighted 2-HSTs, in order to yield an algorithm that is $O(\log n)$-competitive for the Concave MPMD problem on general metrics. We will do so in the following subsection.

\subsection{Matching on Weighted $\sigma$-HSTs}
\label{subsection.metric.matching_on_weighted_HSTs}

Let $T = ( V(T), E(T), w)$ denote an arbitrary weighted $2$-HST defined with respect to a general metric $G = (V(G), E(G), w)$. We denote its height (i.e., the largest distance between $T$’s root and one of its leaves) by $h$. Recall that $h = O(\log n)$ where $n$ denotes the size of the original metric, $G$.

In order to define our algorithm in this case we again use the delay counters $z_1, \ldots, z_d$ as defined in the previous section. However, in this case we add edge counters as well. We denote these counters by $z_e$ for every edge $e \in E(T)$. As with $z_k$, every $z_e$ will have an associated capacity and slope. Its capacity is simply $e$’s weight, $w_e$. Its slope will be defined later on.

Throughout this section we will make heavy use of the set of counters that are descendants of some counter in $F$. Formally, given a counter $z$ let $F_z$ denote the set of all counters that are descendants of $z$ in $F_z$.

As in the single location case, in this case we will associate requests with counters and have requests continually move between counters. Also similar to the earlier case, our algorithm will match any two requests belonging to the same counter. Furthermore, at any point in time, we consider all counters $z$ and
increase them at a rate of $\alpha_z$ if and only if (1) there exists a request associated with $z$ and (2) $F_z$ contains an odd number of requests. Finally, as before, once a counter reaches its capacity, its (only) request is moved to the next counter and the counter is zeroed.

In order to define the way in which the algorithm moves requests between counters we define an Directed Acyclic Graph (DAG) of counters; the nodes represent the counters (both edge and delay counters) and every edge represents the next counter the corresponding request will be moved to once the original counter is filled. Denote by $F(T,D)$ the DAG defined by the metric $T$ and delay function $D$.

\textbf{Defining $F(T,D)$}: Recall that our HST is denoted by $T = (V_T, E_T, w)$. We first iteratively add our delay counters $\{z_k\}_k$ from lowest index to highest, to $T$ as follows. We iterate over every leaf to root path in $T$ and add an edge with weight $y_k$ (i.e., the capacity of $z_k$) to the path between edges $e$ and $e'$ if and only if $w(e) \leq \sum_{j=1}^k y_j < w(e')$ and the delay counter was not already added to this path. Furthermore, we add the new edge below the node connecting $e$ and $e'$ in $T$ (see Figure \ref{figure.metric.adding_delay_counter_to_graph}). Note that we also assume an edge of weight $\infty$ going up from $T$'s root (in order for the process to be well defined). We denote the resulting tree as $F$. Finally, we define the counters (edge or delay) as the bottom nodes of the corresponding edges of $F$.

\begin{figure}[H]
\centering
\includegraphics[width= 8cm]{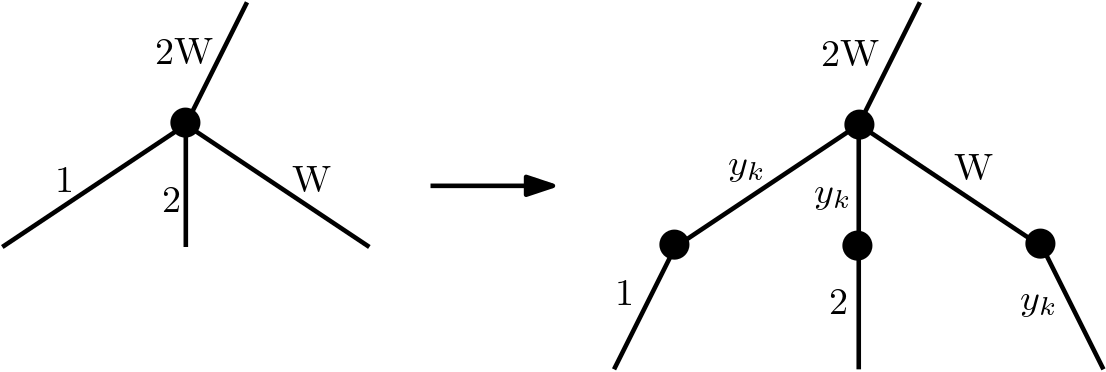}
\caption{Adding delay counter $z_k$ such that $\sum_{j=1}^k y_j = \sqrt{W}$ for some large $W$.}
\label{figure.metric.adding_delay_counter_to_graph}
\end{figure}

We note several properties of $F$. 

\begin{observation}\leavevmode
\label{observation.metric.properties_of_F}
\begin{itemize}
    \item $V_F$ is comprised of delay counters, edge counters (and maybe the root of $V_T$).
    \item If delay counter $z_i$ is a descendant of delay counter $z_j$ in $F$ then $i<j$.
    \item Any path that begins at a leaf and ends at $z_i$ must pass through $z_j$ for any $j < i$.
\end{itemize}
\end{observation}

Now that we have defined $F(T,D)$ we define an edge counter $z_e$'s slope as the slope in the first delay counter encountered on the path from $z_e$ to the root of $F$. Note that since we assumed that $\lim_{t \rightarrow \infty} D(t) = \infty$, we are guaranteed that there always exists such a delay counter.

Recall that for any delay counter $z_k$, its capacity is defined as $y_k = \alpha_k \ell_k$. Therefore, we get the following remark.

\begin{remark}
\label{remark.metric.slope_of_edge_counter}
By the definition of $F$, if an edge counter $z_e$ has a slope of $\alpha_k$ then its capacity  is $y_e = w_e \in [\sum_{i = 1}^{k-1} y_i, \sum_{i = 1}^k y_i)$ if $k > 1$ and $y_e = w_e \in [0, y_1)$ otherwise. 
\end{remark}

\noindent In Algorithm \ref{algorithm.metric} we formally define our algorithm which we denote by Metric-Algorithm (MA) (deferred to the Appendix).

\begin{remark}
\label{remark.MA_terminates}
We note that the algorithm will match all requests since a request increases its counter if the subtree rooted at its counter does not contain any requests. Therefore, eventually all request will move up to $F$'s root (and match) if otherwise unmatched.
\end{remark}

\begin{theorem}
$\algma$ is $(O(1), O(h))$-competitive for any weighted 2-HST, $T$ with height $h$ and any concave delay function, $D$.
\end{theorem}

In order to prove our theorem we follow the proof presented for the single location case: we first upper bound the algorithm's cost by the overall increase in its counters and then upper bound this by the connection and delay costs of any (arbitrary) solution. 

\begin{definition}
For a request $r$ let $z_r'(t)$ be defined as the slope of the counter $r$ belongs to at time $t$ and 0 if it does not belong to any counter.
\end{definition}

\subsubsection{Upper Bounding $\algma$'s Cost}
\label{subsection.metric.upper_bounding_MA_cost}

\begin{proposition}
\label{proposition.metric.bound_algorithm_by_counters}
$\algma \leq 4 \sum_r \int_t z_r'(t)dt$.
\end{proposition}

We define $\delta_t(r), k_t(r), d_r'(t)$ and $s_r(t)$ as in Definitions \ref{definition.single_location.delta_t(r)_k_t(r)} and \ref{definition.single_location.d'_r(t)_s_r(t)}. Note that in the metric case, however, if $r$ is associated with edge counter $z_e$ at time $t$ then $k_t(r)$ is defined as the index of the closest delay counter on the path from $z_e$ to the root. Therefore, $s_r(t) = \alpha_{z_e} = \alpha_{k_t(r)}$. In order to prove Proposition \ref{proposition.metric.bound_algorithm_by_counters} we introduce the following lemma.

\begin{lemma}
\label{lemma.metric.non_real_counter_bounded_by_actual_counter}
At any moment $t$, $\sum_r s_r(t) \leq 2\sum_r z_r'(t)$.
\end{lemma}

\begin{proof}
Let $R$ denote all unmatched requests by our algorithm at time $t$. Recall that by the definition of our algorithm $z_r'(t) \neq s_r(t)$ if and only if $r$ belongs to a counter at time $t$ such that an odd number of requests exist on lower counters (i.e., descendants with respect to $F$). For time $t$, let $R_1$ denote all requests for which $z_r'(t) \neq s_r(t)$ and let $R_2$ denote all requests for which $z_r'(t) = s_r(t)$. Therefore, $\sum_r s_r(t) = \sum_{r \in R_1} s_r(t) + \sum_{r \in R_2} z_r'(t)$ and it is enough to show that $\sum_{r \in R_1} s_r(t) \leq \sum_{r \in R_2} z_r'(t)$.

We define a function $f_t(r) \in R_1 \rightarrow R_2$ using the following process. Let $v$ denote the node in $F$ that represents the counter $r$ belongs to at time $t$. Since $r \in R_1$ we are guaranteed that $v$ has a child $u_1$ (in $F$) that contains an odd number of requests contained within the subtree rooted at $u_1$ (if there is more than one such child, choose arbitrarily). Therefore either (1) $u_1$'s counter has a request belonging to it or (2) $u_1$ has a child $u_2$ that contains an odd number of requests contained within the subtree rooted at $u_2$. If (1) is the case then we define $f_r(t) = u_1$. Otherwise we continue the process iteratively ultimately defining $f_r(t) = u_i$ for some $u_i$ such that $u_i$ contains a request. Note that by the definition of the process $u_i$ will have an even number of requests on its descendant counters (as defined by $F$).

We first note that $f_t(\cdot)$ is clearly well defined. Next we argue that $f_t(\cdot)$ in injective. Assume towards contradiction that $f_t(r) = f_t(r') = u$ for $r \neq r'$. Since $u$ is a descendant of both the counters containing $r$ and $r'$ and since $F$ is a DAG, either $r$ is contained within the path $u \rightarrow r'$ or $r'$ is contained within the path $u \rightarrow r$. If the former is true then $r$ would have been encountered during the process of defining $f_t(r')$ before $u$ and therefore we must have $f_t(r') = u = r$. On the other hand this means that $r \in R_1 \cap R_2$ in contradiction to their definitions. The latter case similarly leads to a contradiction. 

Since $f_t(r)$ belongs to a counter which is a descendant of the counter that $r$ belongs to, we are guaranteed that $s_r(t) \leq s_{f_t(r)}(t)$. Therefore,
\[
\sum_{r \in R_1} s_r(t) \leq \sum_{r \in R_1} s_{f_t(r)}(t) \leq \sum_{r \in R_2} s_r(t) = \sum_{r \in R_2} z_r'(t),
\]
and overall,
\[
\sum_r s_r(t) \leq \sum_{r \in R_1} s_r(t) + \sum_{r \in R_2} z_r'(t) \leq 2\sum_{r \in R} z_r'(t).
\]
\end{proof}





\noindent We are now ready to prove Proposition \ref{proposition.metric.bound_algorithm_by_counters}.

\begin{proof}[Proof of Proposition \ref{proposition.metric.bound_algorithm_by_counters}]

Let $\algma_c$ and $\algma_d$ denote $\algma$'s connection and delay costs respectively. Furthermore let $n_e$ denote the number of times edge $e$ was bought by our algorithm. Therefore, $\algma_c = \sum_e n_e w_e$.

We first observe that due to the definition of our algorithm we connect a request through an edge if and only if the corresponding edge counter was filled (since the request must have passed through that counter). Furthermore, once the request moves to the next counter in the tree, the former counter is emptied. Therefore,
\begin{align}
\label{inequality.metric_connection_cost_bound}
\sum_e n_e w_e \leq \sum_r \int_t z_r'(t)dt.    
\end{align}

Observe that by the definition of $F$ and $k_{m(r)}(r)$, we are guaranteed that any request $r$ must have passed through delay counters $z_1, \ldots, z_{k_{m(r)}(r)-1}$ before being matched. Due to the fact that a counter must be filled in order to be emptied, we are guaranteed that,
\begin{align}
\label{inequality.metric_counter_emptying_bounded_by_actual_counter_increase}
\sum_r \sum_{j=1}^{k_{m(r)}(r) - 1} y_j \leq \sum_r \int_t z_r'(t)dt.
\end{align}

\noindent We note that Lemma \ref{lemma.single_location.real_delay_bounded_by_aux} clearly holds for $\algma$ as well. Therefore,
\begin{align*}
\algma &= \algma_d + \algma_c \leq 
\algma_d + \sum_r \int_t z_r'(t)dt \\ &\leq 
\sum_r \sum_{j=1}^{k_{m(r)}(r) - 1} y_j + \sum_r \int_t s_r(t)dt + \sum_r \int_t z_r'(t)dt \\ &\leq
\sum_r \int_t s_r(t)dt + 2 \sum_r \int_t z_r'(t)dt \leq
4 \sum_r \int_t z_r'(t)dt,
\end{align*}
where the first inequality is due to equation (\ref{inequality.metric_connection_cost_bound}), the second is due to Lemma \ref{lemma.single_location.real_delay_bounded_by_aux}, the third is due to equation (\ref{inequality.metric_counter_emptying_bounded_by_actual_counter_increase}) and the last is due to Lemma \ref{lemma.metric.non_real_counter_bounded_by_actual_counter}.
\end{proof}

\subsubsection{Lower Bounding An Arbitrary Solution's Cost}

For an arbitrary solution to the given instance, we denote by $\algsol_c$ and $\algsol_d$ its connection and delay costs respectively. In this section we will lower bound $\algsol_c$ and $\algsol_d$ by the overall increase in counters by charging this increase to different matchings performed by $\algsol$. The following proposition states this formally.

\begin{proposition}
\label{proposition.metric.bound_counters_by_OPT}
$\sum_r \int_t z_r'(t)dt \leq O(1)\cdot \algsol_c + O(h)\cdot \algsol_d.$
\end{proposition}

The rest of this section is dedicated towards the proof of Proposition \ref{proposition.metric.bound_counters_by_OPT}. Throughout this section, given a counter $z \in V(F)$, we will refer to the set of counters belonging to the tree rooted at $z$  in $F$ as $F_z$ (note that $z \in F_z$).

We charge the increase in counters to $\algsol$ using the same flavor as in the proof of Theorem \ref{theorem.single_location_O(1)_competitive}. Specifically, by splitting the increase in each counter into intervals - recall Definition \ref{definition.single_location.counter_intervals}. Note that in the metric case, an interval ends once the request that is associated with the counter moves to the parent of the counter, as defined by $F$.

We first give several definitions (note that they differ from the earlier case due to the fact that now during a given interval, requests may arrive that we want to ignore; specifically, requests that arrive at points in the metric that are unrelated to the requests we would like to consider).

Throughout the remainder of this section, given a counter $\hat{z}$ and an interval $I_i^{\hat{z}}$ defined with respect to $\hat{z}$, we define $\bm{R(I_i^{\hat{z}})}$ to be the set of all requests that arrived during the time interval $I_i^{\hat{z}}$ and that were given to a counter in $F_{\hat{z}}$ upon arrival.


\begin{definition}
Given an interval $I^{\hat{z}}_i = [t^{\hat{z}}_i, t^{\hat{z}}_{i+1})$ as defined by some counter $\hat{z}$ we denote $|R(I^{\hat{z}}_i)| = \gamma_{\hat{z}}$. We further denote their arrival times by $\{a_j\}_{j=1}^{\gamma_{\hat{z}}}$. Given these notations, if $\gamma_{\hat{z}}$ is odd then we define $I^{odd}_{i,{\hat{z}}} = \cup_{j=1}^{\lfloor \frac{\gamma_{\hat{z}}}{2} \rfloor}[a_{2j-1}, a_{2j}) \cup [a_{\gamma_{\hat{z}}},t_{i+1})$. Otherwise, if $\gamma_{\hat{z}}$ is even we define $I^{odd}_{i,{\hat{z}}} = \cup_{j=1}^{\lfloor \frac{\gamma_{\hat{z}}}{2} \rfloor}[a_{2j-1}, a_{2j})$.
\end{definition}

As in the former section, we refer to $I^{odd}_{i,{\hat{z}}}$ as the odd-subinterval of $I^{\hat{z}}_i$. The following lemma states that any point for which $\hat{z}$ increases must lie within the odd-subinterval of the corresponding interval.

\begin{lemma}
\label{lemma.metric.interval_is_odd}
Consider some time interval $I^{\hat{z}}_i = [t^{\hat{z}}_i, t^{\hat{z}}_{i+1})$ defined with respect to counter $\hat{z}$. Let $t^{\hat{z}} \in I^{\hat{z}}_i$ denote some time for which $\hat{z}$ increases. Therefore, $t \in I_{i, {\hat{z}}}^{odd}$.
\end{lemma}

\noindent The proof is deferred to the Appendix.


\begin{definition}
Given an interval $I$ and a counter $\hat{z}$ we denote the set of requests that arrived during $I$ and were given to counters within $F_{\hat{z}}$ as $R(I, \hat{z})$. Given this notation we say that $\algsol$ is live with respect to $I$ and $\hat{z}$ at time $t\in I$ if $\algsol$'s matching contains an unmatched request $r$ at time $t$ such that either $r \in R(I, \hat{z})$ or $r$'s pair (with respect to $\algsol$) belongs to $R(I, \hat{z})$.
\end{definition}

\noindent The proofs of Lemmas \ref{lemma.metric.last_interval_is_even}, \ref{lemma.metric.parity_of_last_interval_OPT} and \ref{lemma.metric.middle_interval_is_odd} are all deferred to the appendix.

\begin{lemma}
\label{lemma.metric.last_interval_is_even}
For any counter $\hat{z}$, $|R(I^{\hat{z}}_{m_{\hat{z}}-1})|$ is even.
\end{lemma}



\begin{lemma}
\label{lemma.metric.parity_of_last_interval_OPT}
For any counter $\hat{z}$, one of the following conditions hold:
\begin{itemize}
    \item $\algsol$ matched a request from $R(I^{\hat{z}}_{m_{\hat{z}}-1})$ through an edge $e$ such that $z_e$ is an ancestor of $\hat{z}$ in $F$.
    \item If $t \in I_{m_{\hat{z}}-1, {\hat{z}}}^{odd}$ then $\algsol$ must be live with respect to $I^{\hat{z}}_{m_{\hat{z}}-1}$ and $\hat{z}$ at time $t$.
\end{itemize}
\end{lemma}

\begin{lemma}
\label{lemma.metric.middle_interval_is_odd}
For any counter $\hat{z}$, $|R(I^{\hat{z}}_i)|$ is odd.
\end{lemma}

Consider any two consecutive intervals $I_i^{\hat{z}} = [t^{\hat{z}}_i, t^{\hat{z}}_{i+1})$ and $I_{i+1}^{\hat{z}} = [t^{\hat{z}}_{i+1}, t^{\hat{z}}_{i+2})$ such that $t^{\hat{z}}_{i+2} < t^{\hat{z}}_m$ (if such intervals exist), as defined with respect to $\hat{z}$. Let $R(I^{\hat{z}}_i \cup I^{\hat{z}}_{i+1})$ denote the set of requests that arrived during the interval $I^{\hat{z}}_i \cup I^{\hat{z}}_{i+1}$ and that were given to $F_{\hat{z}}$ upon arrival. Lemma \ref{lemma.metric.middle_interval_parity_OPT} will aid us in charging the increase in counters towards $\algsol$ (the proof is deferred to the Appendix).

\begin{lemma}
\label{lemma.metric.middle_interval_parity_OPT}
For $I^{\hat{z}} = I^{\hat{z}}_i$ or $I^{\hat{z}} = I^{\hat{z}}_{i+1}$ one of the following conditions hold:
\begin{itemize}
    \item $\algsol$ matches a request from $R(I^{\hat{z}}_i \cup I^{\hat{z}}_{i+1})$ through an edge $e$ such that $z_e$ is an ancestor of $\hat{z}$ in $F$.
    \item For any $t \in I^{odd}_{\hat{z}}$, $\algsol$ must be live with respect to $I^{\hat{z}}_i \cup I^{\hat{z}}_{i+1}$ and $\hat{z}$ at time $t$.
\end{itemize}
\end{lemma}



\begin{lemma}
\label{lemma.metric.the_number_of_intervals_is_odd}
For any counter $\hat{z}$ the number of intervals defined with respect to that counter, $m_{\hat{z}}$, is odd.
\end{lemma}

\begin{proof}
Follows from Lemmas \ref{lemma.metric.last_interval_is_even} and \ref{lemma.metric.interval_is_odd}.
\end{proof}

\noindent We are now ready to prove Proposition \ref{proposition.metric.bound_counters_by_OPT}.

\begin{proof}[Proof of Proposition \ref{proposition.metric.bound_counters_by_OPT}]
Let $\algsol$ denote an arbitrary matching and let $\algsol_d$ and $\algsol_c$ denote its delay and connection costs. Recall that we aim to prove that $\sum_r \int_t z_r'(t)dt \leq O(1)\cdot \algsol_c + O(h)\cdot \algsol_d$. For a counter $z$ let $z'(t)$ denote the counter's slope if the counter increases at time $t$ and otherwise we define it as 0. Therefore, $\sum_r \int_t z_r'(t)dt = \sum_{\hat{z}} \int_t \hat{z}'(t)dt$.

Recall the definition of $\rho(r)$ (as in the proof of Proposition \ref{proposition.single_location.bound_counters_by_OPT}). Therefore,
\begin{align}
\label{equation.metric.opt_delay_is_counted_twice}
\algsol_d = \frac{1}{2} \sum_r \rho(r).
\end{align}

\noindent Given a request $r$ we denote its connection cost with respect to $\algsol$ as $\kappa(r)$. Therefore,
\begin{align}
\label{equation.metric.opt_connection_is_counted_twice}
\algsol_c = \frac{1}{2} \sum_r \kappa(r).
\end{align}

We say that $r$ is of delay-level (resp. connection-level) $k$ with respect to $\algsol$ if $\rho(r) \in [\sum_{j=1}^k y_j, \sum_{j=1}^{k+1}y_k)$ (resp. $\kappa(r) \in [\sum_{j=1}^k y_j, \sum_{j=1}^{k+1}y_k)$). 

By Lemma \ref{lemma.metric.the_number_of_intervals_is_odd} we may partition our overall set of intervals (defined with respect to $\hat{z}$) into pairs $\{I^{\hat{z}}_{2i}, I^{\hat{z}}_{2i+1} \}$ for $i \in \{0, \ldots, \frac{m_{\hat{z}}-3}{2}\}$ with the addition of $I^{\hat{z}}_{m_{\hat{z}}-1}$ (we abuse notation and let $I_{-1} = \emptyset$ for the case that $m_{\hat{z}}=1$). We charge the increase in our counters to $\algsol$ as follows. We charge the increase in each counter separately and denote our charging scheme by $g_{\hat{z}}(t) = r$ such that $g_{\hat{z}}(\cdot)$ is defined only for points $t$ for which $\hat{z}$ increases. As before we consider two cases: either $t \in I^{\hat{z}}_{m_{\hat{z}}-1}$ or there exists a pair of intervals such that $t \in I^{\hat{z}}_{2i} \cup I^{\hat{z}}_{2i+1}$. 



For any counter $\hat{z}$ and any pair of intervals $\{I^{\hat{z}}_{2i}, I^{\hat{z}}_{2i+1}\}$, we will choose only one of the intervals to charge to $\algsol$ as follows. By Lemma \ref{lemma.metric.middle_interval_parity_OPT} one of these intervals guarantees the condition as defined in them lemma. We assume it is the first (i.e., $I^{\hat{z}}_{2i}$) and use that interval towards our charging scheme (if it were the second, we would have used that interval and continued identically).

We note that for each interval in a pair of intervals $I^{\hat{z}}_{2i} \cup I^{\hat{z}}_{2i+1}$ the increase in $\hat{z}$ is exactly its capacity. Therefore, we may charge the overall increase in the second interval to the first and thereby lose a factor of 2. Therefore, if we let $I^{\hat{z}}_{even} = \cup_{i=0}^{\frac{m_{\hat{z}}-3}{2}}I^{\hat{z}}_{2i} \cup I^{\hat{z}}_{m_{\hat{z}}-1}$ then,
\begin{align}
\label{equation.metric.lose_factor_of_2_by_considering_first_interval_of_pair}
\sum_{\hat{z}} \int_t \hat{z}'(t)dt \leq 2 \sum_{\hat{z}} \int_{t \in I^{\hat{z}}_{even}} \hat{z}'(t)dt.
\end{align}
Therefore, we will define $g_{\hat{z}}(t)$ only if $t \in I^{\hat{z}}_{2i}$ for $i \in [\frac{m_{\hat{z}}-1}{2}]$.

Given the counter $\hat{z}$ we define $k(\hat{z}) \in \mathbbm{N}$ such that $\alpha_{k(\hat{z})}$ denotes the slope of $\hat{z}$. To ease notation in the following definition, let $I_{m_{\hat{z}}}^{\hat{z}} = \emptyset$. We are now ready to formally define our charging scheme.

\noindent \textbf{Defining} $\bm{g_{\hat{z}}(t)}$: For any $i \in \{0,1,\ldots,\frac{m_{\hat{z}}-1}{2}\}$ and any $t \in I^{\hat{z}}_{2i}$ such that $\hat{z}$ increases, we define $g_{\hat{z}}(t)$ as follows. We set $g_{\hat{z}}(t) = r$ such that $r \in R(I^{\hat{z}}_{2i} \cup I^{\hat{z}}_{2i+1})$ and $\algsol$ matches $r$ through an edge $e$ such that $z_e$ is an ancestor of $\hat{z}$ in $F$, if such a request exists. Otherwise, we define $g_{\hat{z}}(t) = r$ such that $r \in R(I^{\hat{z}}_{2i} \cup I^{\hat{z}}_{2i+1})$ and $r$ has a delay level of $\geq k(\hat{z})$ with respect to $\algsol$, if such a request exists. Otherwise, if such a request does not exist, then we will charge the increase to the request $r$ that causes $\algsol$ to be live with respect to $I^{\hat{z}}_{2i} \cup I^{\hat{z}}_{2i+1}$ and $\hat{z}$ at time $t$. \\
Finally, we break ties by always taking the earliest request to arrive (note that any tie breaking that is consistent will suffice).

\begin{remark}
Our charging scheme $g_{\hat{z}}(\cdot)$ is well defined. This is true for any $t \in I^{\hat{z}}_{2i}$ due to the fact that if $\hat{z}$ increases then $t \in I_{2i, {\hat{z}}}^{odd}$ (Lemma \ref{lemma.metric.interval_is_odd}) which in turn guarantees that $\algsol$ is live with respect to $I^{\hat{z}}_{2i}\cup I^{\hat{z}}_{2i+1}$ and $\hat{z}$ at time $t$ (due to Lemma \ref{lemma.metric.middle_interval_parity_OPT} and the fact that we chose the interval that guarantees the defined condition). This is similarly true for any $t \in I^{\hat{z}}_{m_{\hat{z}}-1}$ due to Lemmas \ref{lemma.metric.interval_is_odd} and \ref{lemma.metric.parity_of_last_interval_OPT}.
\end{remark}

Let $D_{\hat{z}}$ denote the set of requests within the image of $g_{\hat{z}}(\cdot)$. Let $A_{\hat{z}}$ denote the set of requests from $D_{\hat{z}}$  that were charged to because they were matched through an edge $e$ such that $z_e$ is an ancestor of $\hat{z}$ in $F$. Let $B_{\hat{z}}$ denote the set of requests from $D_{\hat{z}}$ that were charged to because they had a delay-level that is $\geq k(\hat{z})$. Finally, let $C_{\hat{z}} = D_{\hat{z}} \setminus (A_{\hat{z}} \cup B_{\hat{z}})$.

By changing the summation order we get,
\begin{align}
\label{equation.metric.resumming_over_the_image_of_g}
\sum_{\hat{z}} \int_{t \in I_{even}} \hat{z}'(t)dt &=
\sum_{\hat{z}} \sum_{r \in D_{\hat{z}}} \int_{t \in g^{-1}_{\hat{z}}(r)} \hat{z}'(t)dt =
\sum_r \sum_{\hat{z} : r \in A_{\hat{z}}} \int_{t \in g^{-1}_{\hat{z}}(r)} \hat{z}'(t)dt \nonumber \\&+
\sum_r \sum_{\hat{z} : r \in B_{\hat{z}}} \int_{t \in g^{-1}_{\hat{z}}(r)} \hat{z}'(t)dt +
\sum_r \sum_{\hat{z} : r \in C_{\hat{z}}} \int_{t \in g^{-1}_{\hat{z}}(r)} \hat{z}'(t)dt.
\end{align}

A request $r$ can only belong to a single interval with respect to $\hat{z}$ and therefore since the increase in $\hat{z}$ is at most $y_{\hat{z}}$ in each interval, $r$ may be charged by at most $y_{\hat{z}}$ (exactly $y_{\hat{z}}$ if the interval is not $I_{m-1}$). Therefore, for any request $r$,
\begin{align}
\label{equation.metric.bound_total_charge_in_interval}
\int_{t \in g^{-1}_{\hat{z}}(r)} \hat{z}'(t)dt \leq y_{\hat{z}}.
\end{align}

Consider $\sum_r \sum_{\hat{z}: r \in A_{\hat{z}}} y_{\hat{z}}$. We consider the delay counters and edge counters that satisfy $r \in A_{\hat{z}}$ separately. We first consider the case that $\hat{z} = z_e$ for some edge $e$. Recall that for all edge counters $z_e$, we have that $y_{z_e} = w_e$. If $r \in A_{z_e}$ then $r$ was given to $F_{z_e}$ upon arrival and was matched by $\algsol$ through an edge that is an ancestor of $z_e$ in $F$. Therefore, if $r \in A_{z_e}$ then $r$ was matched through $e$ by $\algsol$. Therefore,
\begin{align}
\label{equation.metric.bound_edge_counters_from_A_to_opt}
\sum_{z_e: r \in A_{z_e}} y_{z_e} = \sum_{e: r \in A_{z_e}} w_e \leq \kappa(r).
\end{align}

Next we consider the case that $\hat{z} = z_k$ (i.e., the delay counters). Let $\bar{k}$ denote the largest $k$ for which $r \in A_{z_k}$ and let $\bar{e}$ denote the largest weighted edge used by $\algsol$ to connect $r$. By the definition of $A_{z_k}$ we are guaranteed that $z_{\bar{e}}$ is an ancestor of $z_{\bar{k}}$ in $F$. Therefore, by the definition of $F$, $\sum_{j=1}^{\bar{k}}y_{z_j} \leq w_{\bar{e}}$, and thus,

By the definition of $F$ we are therefore guaranteed that
\begin{align}
\label{equation.metric.bound_delay_counters_from_A_to_opt}
\sum_{z_k: r \in A_{z_k}} y_{z_k} \leq
\sum_{j=1}^{\bar{k}}y_{z_j} \leq
w_{\bar{e}} \leq
\kappa(r).
\end{align}

Next we consider $\sum_r \sum_{\hat{z}: r \in B_{\hat{z}}} y_{\hat{z}}$ and again consider the delay and edge counters separately. We first consider the case that $\hat{z} = z_e$. Therefore, $y_{z_e} = w_e$. If $r \in B_{z_e}$ then $r$ must have been given to $F_{z_e}$ upon arrival. Therefore, if both $z_e$ and $z_{e'}$ are such that $r \in B_{z_e}$ and $r \in B_{z_{e'}}$ then one must be the ancestor of the other in $F$ and therefore also in $T$. Since $T$ is a 2-HST we are guaranteed that $\sum_{z_e: r \in B_{z_e}}w_e \leq 2 w_{\bar{e}}$ where $\bar{e}$ denotes such an edge that is closest to the root of $F$ (equivalently $T$). On the other hand, due to the fact that $r \in B_{z_{\bar{e}}}$ we are guaranteed by the construction of $F$ that $w_{\bar{e}} \leq \sum_{j=1}^{k(z_e)}y_{z_j} \leq \rho(r)$. Therefore, 
\begin{align}
\label{equation.metric.bound_edge_counters_from_B_to_opt}
\sum_{z_e: r \in B_{z_e}} y_{z_e} = 
\sum_{e: r \in B_{z_e}} w_e \leq 
2w_{\bar{e}} \leq 
2\sum_{j=1}^{k(z_e)}y_{z_j} \leq
2\rho(r).
\end{align}

We now consider the case that $\hat{z} = z_k$ (i.e., the delay counters). By the definition of $B_{z_k}$ we are guaranteed that $r$ has delay-level $\geq \bar{k}$ where $\bar{k}$ denotes the closest delay counter to the root from all delay counters such that $r \in B_{z_k}$. Therefore, $\sum_{j=1}^{\bar{k}} y_{z_j} \leq
\rho(r)$. Again, if $r \in B_{z_k}$ then $r$ must have been given to $F_{z_k}$ upon arrival. Therefore, if $z_k$ and $z_{k'}$ are both such counters, then one must be the ancestor of the other in $F$, resulting in $\sum_{z_k: r \in B_{z_k}} y_{z_k} \leq
\sum_{j=1}^{\bar{k}} y_{z_j}$. Therefore, overall,
\begin{align}
\label{equation.metric.bound_delay_counters_from_B_to_opt}
\sum_{z_k: r \in B_{z_k}} y_{z_k} \leq
\sum_{j=1}^{\bar{k}} y_{z_j} \leq
\rho(r).
\end{align}

Finally we consider $\sum_r \sum_{\hat{z}: r \in C_{\hat{z}}} \int_{t \in g^{-1}_{\hat{z}}(r)} \hat{z}'(t)dt$. Again, if $r \in C_{\hat{z}}$ then $r$ must have been given to $F_{\hat{z}}$ upon arrival. Therefore, if $\hat{z}$ and $\hat{z}'$ are both such counters, then one must be the ancestor of the other in $F$. Furthermore, if $\hat{z}$ and $\hat{z}'$ are both delay counters then their slopes must be different. Thus, due to the fact that the delay counters' slopes decrease exponentially and there are at most $h$ edge counters on a leaf to root path in $F$ we are guaranteed that for any time $t$, $\sum_{\hat{z}: r \in C_{\hat{z}}} \hat{z}'(t) \leq 2h (\bar{z}(t))'(t)$, where $\bar{z}(t)$ denotes the counter with the largest slope taken from the set of all counters satisfying $r \in C_{\hat{z}}$. Therefore,


\begin{align}
\label{equation.metric.bound_by_largest_slope_counter}
\sum_{\hat{z}: r \in C_{\hat{z}}} \int_{t \in g^{-1}_{\hat{z}}(r)} \hat{z}'(t)dt \leq 2h \int_{t \in \cup_{\hat{z}: r \in C_{\hat{z}}} g^{-1}_{\hat{z}}(r)}(\bar{z}(t))'(t)dt.
\end{align}

By the definition of $g_{\hat{z}}(\cdot)$ we are guaranteed that either $r$ is unmatched at time $t$ or there exists an unmatched request that will be matched to $r$ in the future. Furthermore, by the definition of $C_{\hat{z}}$ we are guaratneed that $r$'s momentary delay as incurred by $\algsol$ is at least $\alpha_{k(\hat{z})}$ for any $\hat{z} \in C_{\hat{z}}$ (and in particular $\bar{z}(t)$). Therefore,  
\begin{align}
\label{equation.metric.opt_pays_a_momentary_cost_that_is_atleast_alg}
\int_{t \in \cup_{\hat{z}: r \in C_{\hat{z}}} g^{-1}_{\hat{z}}(r)} (\bar{z}(t))'(t)dt = \int_{t \in \cup_{\hat{z}: r \in C_{\hat{z}}} g^{-1}_{\hat{z}}(r)} \alpha_{k(\bar{z}(t))} \leq \rho(r).
\end{align}

\noindent Combining equations (\ref{equation.metric.bound_by_largest_slope_counter}) and (\ref{equation.metric.opt_pays_a_momentary_cost_that_is_atleast_alg}),
\begin{align}
\label{equation.metric.bound_counters_from_C_to_opt}
\sum_{\hat{z}: r \in C_{\hat{z}}} \int_{t \in g^{-1}_{\hat{z}}(r)} \hat{z}'(t)dt \leq 2h \cdot \rho(r).
\end{align}

\noindent Therefore, overall we get,
\begin{align*}
\sum_r \int_t z_r'(t)dt  &= 
\sum_{\hat{z}} \int_t \hat{z}'(t)dt \leq 
2 \sum_{\hat{z}} \int_{t \in I_{even}} \hat{z}'(t)dt \leq 
\sum_{\hat{z}} \sum_{r \in D_{\hat{z}}} \int_{t \in g^{-1}_{\hat{z}}(r)} \hat{z}'(t)dt \\&=
\sum_r \sum_{\hat{z} : r \in A_{\hat{z}}} \int_{t \in g^{-1}_{\hat{z}}(r)} \hat{z}'(t)dt \nonumber +
\sum_r \sum_{\hat{z} : r \in B_{\hat{z}}} \int_{t \in g^{-1}_{\hat{z}}(r)} \hat{z}'(t)dt \nonumber +
\sum_r \sum_{\hat{z} : r \in C_{\hat{z}}} \int_{t \in g^{-1}_{\hat{z}}(r)} \hat{z}'(t)dt \\ &\leq
\sum_r \sum_{\hat{z} : r \in A_{\hat{z}}} y_{\hat{z}} \nonumber +
\sum_r \sum_{\hat{z} : r \in B_{\hat{z}}} y_{\hat{z}} \nonumber +
\sum_r \sum_{\hat{z} : r \in C_{\hat{z}}} \int_{t \in g^{-1}_{\hat{z}}(r)} \hat{z}'(t)dt \\ &\leq
\sum_r 2\kappa(r) + \sum_r \sum_{\hat{z} : r \in B_{\hat{z}}} y_{\hat{z}}  + \sum_r \sum_{\hat{z} : r \in C_{\hat{z}}} \int_{t \in g^{-1}_{\hat{z}}(r)} \hat{z}'(t)dt \\ &\leq 
\sum_r 2\kappa(r) + \sum_r 3 \rho(r)  + \sum_r \sum_{\hat{z} : r \in C_{\hat{z}}} \int_{t \in g^{-1}_{\hat{z}}(r)} \hat{z}'(t)dt \\ &\leq 
\sum_r 2\kappa(r) + \sum_r 3 \rho(r) + \sum_r 2h \cdot \rho(r) \\ &\leq
4 \cdot \algsol_c + (4h + 6) \cdot \algsol_d,
\end{align*}
where the equalities are simply through a change of summation order, the first inequality is due to equation (\ref{equation.metric.lose_factor_of_2_by_considering_first_interval_of_pair}), the second inequality is due to equation (\ref{equation.metric.bound_total_charge_in_interval}), the third inequality is due to equations (\ref{equation.metric.bound_edge_counters_from_A_to_opt}) and (\ref{equation.metric.bound_delay_counters_from_A_to_opt}), the fourth inequality is due to equations (\ref{equation.metric.bound_edge_counters_from_B_to_opt}) and (\ref{equation.metric.bound_delay_counters_from_B_to_opt}), the fifth inequality is due to equation (\ref{equation.metric.bound_counters_from_C_to_opt}) and the last inequality is due to equations (\ref{equation.metric.opt_delay_is_counted_twice}) and (\ref{equation.metric.opt_connection_is_counted_twice}).
\end{proof}

\noindent Combining Propositions \ref{proposition.metric.bound_algorithm_by_counters} and \ref{proposition.metric.bound_counters_by_OPT} yields the following theorem.

\begin{theorem}
\label{theorem.MA_competitive_on_HST}
$\algma$ is $(O(1), O(h))$-competitive for any HST with height $h$ and any concave delay function.
\end{theorem}

\noindent Finally, combining Theorem \ref{theorem.MA_competitive_on_HST} with Lemmas \ref{lemma.embed_metric_to_hst} and \ref{lemma.transfer_competitiveness_from_hst_to_general_metric}, yields the following theorem.

\begin{theorem}
$\algma$ is $O(\log n)$-competitive for any metric and any concave delay function.
\end{theorem}

\section{Single Location Bipartite Matching}
\label{section.bipartite.single_location}
In this section we consider the Single Location Concave MBPMD problem. We will ultimately show an $O(1)$-competitive algorithm for this problem. Our algorithm will make use of counters for each linear piece in the delay function denoted by $z_1^+, z_1^-, \ldots, z_d^+, z_d^-$. Every counter will have a slope and a capacity which is defined to be $y_k$ and $\alpha_k$ for counters $z_k^+$ and $z_k^-$. Before formally defining our algorithm we need the following definition.

\begin{definition}
Given counters $z_k^+$ and $z_k^-$ we denote by $P_k(t)$ (resp. $N_k(t)$) the number of positive (resp. negative) requests associated with counter $z_k^+$ (resp. $z_k^-$) at time $t$. Finally, we define the surplus of the prefix of counters as $sur_k(t) = \sum_{i=1}^k (P_i(t) - N_i(t))$.
\end{definition}

Our algorithm is defined as follows. We will associate positive requests with positive counters and negative requests with negative counters. Once a request arrives, we associate it with $z_1$ of the corresponding polarity. If there are 2 requests of opposite polarity on the same indexed counters, match them. Otherwise, at any point in time, we consider all counters $z_k$ simultaneously and increase counter $z_k^+$ (resp. $z_k^-$) at a rate of $\alpha_k|sur_k(t)|$ if and only if there is at least one request associated with the counter and $sur_k(t) > 0$ (resp. $sur_k(t) < 0$). Finally, if any counter $z_k^+$ (resp. $z_k^-$) reaches its capacity, move a single request that is associated with it to $z_{k+1}^+$ (resp. $z_{k+1}^-$) and reset both $z_k^+$ and $z_k^-$ to 0. The algorithm is formally defined in Algorithm \ref{algorithm.bipartite.single_location.BPSLA} (deferred to the Appendix).


\begin{theorem}
\label{theorem.bipartite.single_location.BSLA_is_constant_competitive}
$\algbpsla \leq O(1) \cdot \optim$.
\end{theorem}

In order to prove our theorem, we first introduce the following definition and then bound our algorithm's cost: we first bound the algorithm's cost by the increase in its counters and then bound the increase in its counters by the cost of the optimal solution.

\begin{definition}
Given counters $z_k^+$ and $z_k^-$ define $z_k'(t)$ to be $\alpha_k |sur_k(t)|$ if there exists a positive (resp. negative) request on $z_k^+$ (resp. $z_k^-$) and $sur_k(t) > 0$ (resp. $sur_k(t) < 0$). Otherwise, $z_k'(t) = 0$.
\end{definition}

\subsection{Upper Bounding $\algbpsla$'s Cost}

In this section we would like to upper bound $\algbpsla$ by the overall increase in its counters. In order to do so recall the definitions of $\delta_t(r), k_t(r), d'_r(t)$ and $s_r(t)$ as defined in Definitions \ref{definition.single_location.delta_t(r)_k_t(r)} and \ref{definition.single_location.d'_r(t)_s_r(t)}.



\begin{proposition}
\label{proposition.bipartite.single_location.bound_BPSLA_by_counters}
$\algbpsla \leq 3 \sum_k \int_t z_k'(t)dt$.
\end{proposition}

\noindent Before proving our propositions we introduce the following lemmas.

\begin{lemma}
\label{lemma.bipartite.single_location.bound_srt_by_increase_in_counters}
$\sum_r \int_t s_r(t)dt \leq 2 \sum_k \int_t z_k'(t)dt.$
\end{lemma}

\begin{proof}
To prove our lemma we will charge the value $\sum_r s_r(t)$ for a given time $t$ using a charging scheme $f_t$. Note that the scheme is defined for a given time $t$ and may change over time.

$f_t$ is defined as follows. We iterate over the counters containing requests from lowest to highest. We begin such that all requests are unmarked and we will mark them as we iterate over the counters. 

Consider the iteration in which the counter $z_{k}$ was encountered. Let $A_{k}$ denote the requests belonging to counter $z_{k}$ (as we will see, we only mark requests from lower counters and since we are iterating over the counters, lowest to highest, we are guaranteed that $A_{k}$ are unmarked). W.l.o.g. assume that they are positive. Let $B_{k}$ denote the set of unmarked negative requests associated with counters $z_1, \ldots, z_{k-1}$ (note that $B_{k}$ might be empty). Consider an arbitrary subset $A \subset A_{k}$ of size $\min \{|A_{k}|, |B_{k}|\}$. Define $f_t$ on $A$ as a (arbitrary) 1-1 mapping to $B_{k}$. Mark all requests in $B_{k}$ and $A$.

The process partitions all the requests into three sets: the domain of $f_t$, denoted by $\mathcal{D}_t$, the image of $f_t$, denoted by $f_t(\mathcal{D}_t)$ and all the rest. We note that both $\mathcal{D}_t$ and $f_t(\mathcal{D}_t)$ may both contain positive and negative requests simultaneously.

By the definition of $f_t$, $f_t: \mathcal{D}_t \rightarrow f(\mathcal{D}_t)$ is 1-1 and always maps requests to requests with opposite polarity. Recall the definitions of $P_k(t) = P_k$ and $N_k(t) = N_k$. Note that since our algorithm matches requests of opposite polarity that are associated with the same counter, $\min \{P_k, N_k\} = 0$. We say that $r \in z_k$ if $r$ is associated with $z_k$ at time $t$. 

We first consider any counter $z_k$ with $z_k'(t) \neq 0$. Assume w.l.o.g. that the requests are positive. Therefore, by the definition of our algorithm $sur_k(t) > 0$. Observe that by the definition of $f_t$, $|\{r: r \in \mathcal{D}_t \land r \in z_k\}| = \min \{P_k, \sum_{j=1}^{k-1} N_j\}$. By the definition of $P_k$, $|\{r: r \in \mathcal{D}_t \land r \in z_k\}| + |\{r: r \not \in \mathcal{D}_t \land r \in z_k\}| = P_k$. Therefore, if $P_k \geq \sum_{j=1}^{k-1} N_j$ then, 
\begin{align*}
|\{r: r \not \in \mathcal{D}_t \land r \in z_k\}| &= 
P_k - \min \{P_k, \sum_{j=1}^{k-1} N_j \} = 
P_k - \sum_{j=1}^{k-1} N_j \leq 
\sum_{j=1}^{k} P_j - \sum_{j=1}^{k-1} N_j = 
sur_k(t).    
\end{align*}
On the other hand, if $P_k \leq \sum_{j=1}^{k-1} N_j$ then,
\begin{align*}
|\{r : r \not \in \mathcal{D}_t \land r \in z_k\}| = 
0 \leq 
sur_k(t).
\end{align*}
Therefore, in any case, by summing over all counters we are guaranteed that,
\begin{align}
\label{equation.bipartite.single_location.bound_srt_by_increase_in_counters_1}
\sum_{r \not \in \mathcal{D}_t} s_r(t) &=
\sum_{r \not \in \mathcal{D}_t \land \forall k: r \not \in z_k} s_r(t) + \sum_{k} \sum_{r \not \in \mathcal{D}_t \land r \in z_k}s_r(t) \nonumber \\ &=
\sum_{k} \sum_{r \not \in \mathcal{D}_t \land r \in z_k}s_r(t) =
\sum_{k} \alpha_k |\{r: r \not \in \mathcal{D}_t \land r \in z_k\}|   \leq
\sum_k \alpha_k |sur_k(t)|,
\end{align}
where the second equality is due to the fact that $s_r(t) = 0$ for requests that are not associated with any counter at time $t$, the second equality follows from the definition of $s_r(t)$ and the second inequality follows from our earlier discussion.

Observe that due to the fact that $f_t$ maps requests to requests associated with counters of strictly lower levels, we are guaranteed that $s_r(t) \leq s_{f_t(r)}(t)$ for all $r \in \mathcal{D}_t$. Therefore, 
\begin{align}
\label{equation.bipartite.single_location.bound_srt_by_increase_in_counters_2}
\sum_{r \in \mathcal{D}_t} s_r(t) \leq
\sum_{r \in \mathcal{D}_t} s_{f_t(r)}(t)  =
\sum_{r \in f_t(\mathcal{D}_t)} s_r(t) \leq
\sum_{r \not \in \mathcal{D}_t} s_r(t),
\end{align}
where the equality is due to the fact that $f_t$ is 1-1 and the second inequality is due to the fact that $\mathcal{D}_t \cap f_t(\mathcal{D}_t) = \emptyset$.

\noindent Combining the above and summing over all points in time,
\begin{align*}
\int_t \sum_r s_r(t)dt &= 
\int_t \sum_{r \in \mathcal{D}_t} s_r(t)dt + \int_t \sum_{r \not \in \mathcal{D}_t} s_r(t)dt \\&\leq
2 \int_t \sum_{r \not \in \mathcal{D}_t} s_r(t)dt \leq
2 \int_t \sum_k \alpha_k |sur_k(t)|dt =
2 \int_t \sum_k z_k'(t) dt.
\end{align*}
where the first inequality is due to equation (\ref{equation.bipartite.single_location.bound_srt_by_increase_in_counters_2}) and the second is due to equation (\ref{equation.bipartite.single_location.bound_srt_by_increase_in_counters_1}).
\end{proof}

\begin{proof}[Proof of Proposition \ref{proposition.bipartite.single_location.bound_BPSLA_by_counters}]
We observe that due to the fact that once a request moves up a counter, the former counter's value is reset to 0, we have,
\begin{align}
\label{equation.single_location.bound_BPSLA_by_counters_2}
\sum_r  \sum_{k = 1}^{k_{m(r)}(r)-1} y_k \leq \sum_k \int_t z'_k(t)dt,
\end{align}
since the capacity of every counter is $y_k$. We note that Lemma \ref{lemma.single_location.real_delay_bounded_by_aux} clearly holds for $\algbpsla$ as well. Therefore,
\begin{align*}
\algbpsla &= 
\sum_r \int_t d_r'(t)dt \leq
\sum_r  \sum_{k = 1}^{k_{m(r)}(r)-1} y_k  + \sum_r  \int_t s_r(t)dt \nonumber \\&\leq
\sum_k \int_t z'_k(t)dt + \sum_r  \int_t s_r(t)dt \leq
3 \sum_k \int_t z_k'(t)dt,
\end{align*}
where the first inequality follows from Lemma \ref{lemma.single_location.real_delay_bounded_by_aux}, the second inequality follows from equation (\ref{equation.single_location.bound_BPSLA_by_counters_2}) and the third inequality follows from Lemma \ref{lemma.bipartite.single_location.bound_srt_by_increase_in_counters}.
\end{proof}

\subsection{Lower Bounding $\optim$'s Cost}

In this section we upper bound the overall increase in counters by the optimal solution.

\begin{proposition}
\label{proposition.bipartite.single_location.bound_counters_by_opt}
$\sum_k \int_t z_k'(t)dt \leq 12\cdot \optim$.
\end{proposition}

\noindent We introduce several definitions to aid us in our proof.


\noindent Recall that $[x_{i-1}, x_{i})$ is defined as the $i$'s linear step in our delay cost function.

\begin{definition}
\label{definition.bipartite.single_location.delta^*_t(r)}
Given a request $r$ with arrival time $a(r)$ and some time $t$ such that $r$ is unmatched at time $t$ with respect to $\optim$, define $\delta^*_t(r)$ such that $t - a(r) \in [x_{\delta^*_t(r)-1}, x_{\delta^*_t(r)})$.
\end{definition}

\begin{definition}
Let $o_k(t)$ denote the number of unmatched requests $r$ at time $t$ with respect to $\optim$ such that $\delta^*_t(r) \leq k$. Further, let $sur^*_k(t)$ denote the positive surplus these requests.
\end{definition}

\noindent Note that for every pair of matched requests, with respect to $\optim$'s matching, only one incurs delay. 

\begin{definition}
\label{definition.bipartite.single_location.E_^k_I}
For a given time interval $I = [a,b)$, let $\mathcal{E}^k_{I}$ denote the number of requests $r$, such that either (1) $\delta^*_t(r)$ changed from $k$ to $k+1$ at time $t \in I$ or (2) $r$ is matched (upon arrival) at time $t\in I$ by $\optim$ to a request $r'$ such that $\delta^*_t(r') \geq k$.
\end{definition}

\noindent For a pictorial example of Definition \ref{definition.bipartite.single_location.E_^k_I} see Figure \ref{figure.bipartite.single_location.E^k_I}.

\begin{figure}[H]
\centering
\includegraphics[width= 11cm]{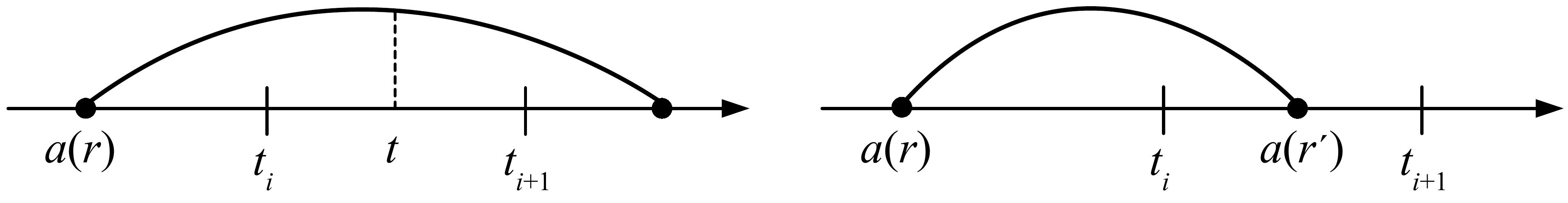}
\caption{In the left hand side, $r$ is counted towards $\mathcal{E}^k_{I}$ due to condition (1) of Definition \ref{definition.bipartite.single_location.E_^k_I}. In the right hand side, $r$ is counted towards $\mathcal{E}^k_{I}$ due to condition (2) of Definition \ref{definition.bipartite.single_location.E_^k_I}.}
\label{figure.bipartite.single_location.E^k_I}
\end{figure}

As in \cite{Min_Cost_Bipartite_Perfect_Matching_with_Delays} we define a potential function to aid in our proof. Specifically, let $\phi_k(t) = 2 y_k |sur_k(t) - sur^*_k(t)|$. Furthermore, given counters $z_k^+$ and $z_k^-$ we partition our time axis into intervals $[t^k_i, t^k_{i+1})$, as in Definition \ref{definition.single_location.counter_intervals}. Note that here, however, an interval ends once either a positive or negative request moves up a counter. We then define $\Delta_i(\phi_k(t)) = \phi_k(t^k_{i+1}) - \phi_k(t^k_i)$. Lemma \ref{lemma.bipartite.single_location.yossis_lemma} follows similarly to the proof of Lemma 19 in \cite{Min_Cost_Bipartite_Perfect_Matching_with_Delays} while integrating the notions as defined in the problem at hand (i.e., Concave Single Location MBPMD).

\begin{lemma}
\label{lemma.bipartite.single_location.yossis_lemma}
Let $I^k_i$ denote some phase interval $[t^k_i, t^k_{i+1})$ defined with respect to counter $z_k$. Then, $\int_{t \in I^k_i}z_k'(t)dt + \Delta_i(\phi_k (t)) \leq 4\alpha_k \cdot \int_{t \in I^k_i} o_k(t) dt + 2\mathcal{E}^k_{I^k_i}y_k$.
\end{lemma}

\begin{proof}
Let $\mathcal{E} = \mathcal{E}^k_{I^k_i}$ and $I^k_i = I_i$. Observe that $f_k(t) = sur^*_k(t) - sur_k(t)$ can change if and only if: (1) $\algbpsla$ moves a request from $z_k$ to $z_{k+1}$, (2) $\delta^*_t(r)$ increased from $k$ to $k+1$ at time $t \in I_i$ or (3) $r$ is matched (upon arrival) at time $t\in I_i$ by $\optim$ to a request $r'$ such that $\delta^*_t(r') \geq k$.

A change of type (1) can only happen once during the phase and changes of types (2) and (3) happen exactly $\mathcal{E}$ times. Therefore, $|\Delta_i(f_k(t))| \leq \mathcal{E}+1$. Note that always, $\Delta_i|f_k(t)| \leq |\Delta_i(f_k(t))|$ and thus overall,

\begin{align}
\label{equation.bipartite.single_location.yossis_lemma_1}
\Delta_i|f_k(t)| \leq |\Delta_i(f_k(t))| \leq \mathcal{E}+1.
\end{align}

We first consider the case that the interval is not last (i.e., $I_i  = [t_i, t_{i+1}) \neq [t_{m-1}, t_m = \infty)$). Assume w.l.o.g. that the phase ends once the algorithm moves a positive request from $z_k$ to $z_{k+1}$ (the second case is proven symmetrically). We consider two cases, either $\Delta_i|f_k(t)| = \mathcal{E}+1$ or not. We first assume that it is the former case. 

\begin{itemize}
    \item $\bm{\Delta_i|f_k(t)| = |\Delta_i(f_k(t))| = \mathcal{E}+1}$: As noted earlier there are only 3 events that may cause $f_k(t)$ to change. By our assumption that the phase ends once the algorithm moves a positive request from $z_k$ to $z_{k+1}$ we are guaranteed that the event of type (1) causes $f_k(t)$ to increase by 1. Since there are at most $\mathcal{E}+1$ events that may cause $f_k(t)$ to change and due to the fact that $|\Delta_i(f_k(t))| = \mathcal{E}+1$ we are guaranteed that $f_k(t)$ may only increase during the phase and that $\Delta_i(f_k(t)) = \mathcal{E}+1$. The following observation is always true.
    \begin{observation}
    $\Delta_i|f_k(t)| = |\Delta_i(f_k(t))| \Rightarrow sign(f_k(t_i)) = sign(\Delta_i(f_k(t)))$.
    \end{observation}
    Therefore, $f_k(t_i) \geq 0$ and throughout the phase we have that $sur_k^*(t) \geq sur_k(t)$. Thus, whenever $z_k^+(t)$ increases, we are guaranteed that $sur_k^*(t) \geq sur_k(t) > 0$. Note that always $o_k(t) \geq sur^*_k(t)$ and that whenever $z_k^+$ increases, it does so at a rate of $\alpha_k \cdot sur_k(t)$. Therefore, whenever $z_k^+$ increases we are guaranteed that,
    \begin{align*}
    o_k(t) \geq sur^*_k(t) \geq sur_k(t) >0.
    \end{align*}
    
    Recall that the phase ended once a positive request left the counter. Therefore, the overall increase in $z_k^+$ is exactly $y_k$. Thus, if we let $\mathcal{J}$ denote all times for which $z_k^+$ increased, then,
    \begin{align}
    \label{equation.bipartite.single_location.yossis_lemma_2}
    y_k =
    \int_{t \in \mathcal{J}} \alpha_k sur_k(t)dt \leq
    \int_{t \in \mathcal{J}} \alpha_k o_k(t)dt \leq 
    \int_{t \in \mathcal{J}} \alpha_k o_k(t)dt,
    \end{align}
    where the last inequality follows from the positivity of $\alpha_k o_k(t)$.

    Since $\Delta_i(\phi_k) = 2y_k\Delta_i(|f_k(t)|)$, we are guaranteed that,
    \begin{align*}
    \int_{t \in I_i}z_k'(t)dt + \Delta_i(\phi_k) \leq
    2y_k + \Delta_i(\phi_k) = 
    2y_k + 2y_k(\mathcal{E}+1) \leq 
    4\int_{t \in I_i}\alpha_k o_k(t)dt + 2y_k \mathcal{E},
    \end{align*}
    where the first inequality follows due to the fact that during a phase $z_k'(t)$ may increase by at most $2y_k$ ($y_k$ for $z_k^+$ and at most $y_k$ for $z_k^-$) and the second inequality follows from equation (\ref{equation.bipartite.single_location.yossis_lemma_2}).
    
    \item $\bm{\Delta_i|f_k(t)| < \mathcal{E}+1}$: Due to the fact that $f_k(t)$ changes value exactly $\mathcal{E}+1$ times and therefore changes parity $\mathcal{E}+1$ times, we are guaranteed that in fact $\Delta_i|f_k(t)| < \mathcal{E}$. Therefore, since $\Delta_i(\phi_k) = 2y_k\Delta_i(|f_k(t)|)$,
    
    \begin{align*}
    \int_{t \in I_i}z_k'(t)dt + \Delta_i(\phi_k) \leq 
    2y_k + \Delta_i(\phi_k) \leq 
    2y_k + 2y_k(\mathcal{E}-1) \leq 
    4\int_{t \in I_i}\alpha_k o_k(t)dt + 2 y_k \mathcal{E},
    \end{align*}
    where the last inequality follows simply from $\alpha_k o_k(t)$'s positivity.
\end{itemize}

\noindent Now that we have proven the lemma for any phase other than the last, we consider the case that $I_{m-1} = [t_{m-1}, t_m = \infty)$. At the end of the phase we have that $sur^*_k(t_m) = sur_k(t_m)$ and therefore, $\Delta_{m-1}(\phi_k) \leq 0$. If $\mathcal{E} \geq 1$ then we are guaranteed that,
\begin{align*}
\int_{t \in I_{m-1}}z_k'(t)dt + \Delta_{m-1}(\phi_k) \leq 2y_k \leq 2y_k\mathcal{E} \leq 4\int_{t \in I_{m-1}}\alpha_k o_k(t)dt + 2 y_k \mathcal{E},
\end{align*}
where the last inequality follows simply from $\alpha_k o_k(t)$'s positivity.

If $\mathcal{E} = 0$ then $sur^*_k(t)$ remains constant throughout the phase and therefore $\phi_k(t) = 0$ throughout the phase. Therefore, $sur^*_k(t) = sur_k(t)$ whenever $z_k^+$ or $z_k^-$ increase. Furthermore, each of the requests counted by $o_k(t)$ incurs a momentary delay of at least $\alpha_k$. Since $|sur^*_k(t)| \leq o_k(t)$ and $z_k'(t) = \alpha_k|sur_k(t)|$ we are guaranteed that in fact $\int_{t \in I_{m-1}}z_k'(t)dt \leq \int_{t \in I_{m-1}}\alpha_k o_k(t)dt$. Therefore,
\begin{align*}
\int_{t \in I_{m-1}}z_k'(t)dt + \Delta_{m-1}(\phi_k) = \int_{t \in I_{m-1}}z_k'(t)dt \leq \int_{t \in I_{m-1}}\alpha_k o_k(t)dt \leq 4\int_{t \in I_{m-1}}\alpha_k o_k(t)dt + 2 y_k \mathcal{E},
\end{align*}
where the last inequality follows from the positivity of $\alpha_k o_k(t)$ and since $\mathcal{E} = 0$. This proves the lemma for the last phase.
\end{proof}

\noindent We are now ready to prove Proposition \ref{proposition.bipartite.single_location.bound_counters_by_opt}.

\begin{proof}[Proof of Proposition \ref{proposition.bipartite.single_location.bound_counters_by_opt}]
By Lemma \ref{lemma.bipartite.single_location.yossis_lemma}, due to the fact that $\phi = 0$ at the beginning and end of the instance, if we sum over all phases of counter $z_k$ we get that,
\begin{align}
\label{equation.bipartite.single_location.yossis_lemma.1}
\int_{t}z_k'(t)dt  \leq 4\int_{t}\alpha_k o_k(t)dt + 2 y_k \mathcal{E}_k,
\end{align}
where $\mathcal{E}_k = \sum_{i}\mathcal{E}^k_{I_i}$. In fact, $\mathcal{E}_k$ is the set of all requests with $\rho(r) \geq \sum_{i=1}^k y_i$ (where $\rho(r)$ is defined as in the proof of Proposition \ref{proposition.single_location.bound_counters_by_OPT}). We first argue that $\sum_{k} y_k \mathcal{E}_k \leq 2 \optim$. Indeed, by changing the order of summation,
\begin{align}
\label{equation.bipartite.single_location.yossis_lemma.2}
\sum_{k} y_k \mathcal{E}_k &= 
\sum_{k} y_k \sum_r \mathbbm{1}\{\rho(r) \geq \sum_{i=1}^k y_i\} \nonumber \\ &=
\sum_r \sum_{k} y_k \mathbbm{1}\{\rho(r) \geq \sum_{i=1}^k y_i\} \leq
\sum_r \rho(r) \leq 
2 \cdot \optim.
\end{align}

Next we argue that $\sum_{k} \int_t \alpha_k o_k(t)dt \leq 2 \optim$. Recall that $\delta^*_t(r)$ such that request $r$ incurs a delay of $\alpha_{\delta^*_t(r)}$ by $\optim$'s matching at time $t$. By changing the order of summation and relying on the fact that $\alpha_k$ decrease exponentially we are guaranteed that,
\begin{align}
\label{equation.bipartite.single_location.yossis_lemma.3}
\sum_{k} \int_t \alpha_k o_k(t)dt &= 
\sum_{k} \int_t \alpha_k \sum_r \mathbbm{1}\{\delta^*_t(r) \leq k\} =
\int_t \sum_r \sum_{k} \alpha_k \mathbbm{1}\{\delta^*_t(r) \leq k\} \nonumber \\&=
\int_t \sum_r \sum_{k \geq \delta^*_t(r)} \alpha_k \leq
\int_t \sum_r 2\alpha_{\delta^*_t(r)} =
2 \cdot \optim.
\end{align}

\noindent Combining the above yields,
\begin{align*}
\sum_k \int_t z_k'(t)dt \leq 
4 \sum_k \int_{t}\alpha_k o_k(t)dt + 2 \sum_k y_k \mathcal{E}_k \leq
12 \cdot \optim.
\end{align*}
\end{proof}

\noindent We are now ready to prove Theorem \ref{theorem.bipartite.single_location.BSLA_is_constant_competitive}.

\begin{proof}[Proof of Theorem \ref{theorem.bipartite.single_location.BSLA_is_constant_competitive}]
Combining Propositions \ref{proposition.bipartite.single_location.bound_BPSLA_by_counters} and \ref{proposition.bipartite.single_location.bound_counters_by_opt} yields the theorem.
\end{proof}

\section{Bipartite Matching on a Metric}
\label{section.bipartite.metric}
In this section we consider the online min-cost perfect bipartite matching with concave delays problem given a general metric (Concave MBPMD). Recall that in this problem the cost of a matching is comprised of both a connection cost and a delay cost. We solve this problem by first reducing it to metrics based on weighted $\sigma$-HSTs and then define an algorithm for such a case.

As in Subsection \ref{subsection.metric.reduction_to_HSTs}, in the bichromatic case proving a algorithm is $(O(1), O(h))$-competitive on metrics defined by weighted 2-HSTs with height $h$, yields a randomized algorithm that is $O(\log n)$ for the Concave MBPMD problem on general metric. The rest of this section is devoted to proving the following theorem.

\begin{theorem}
\label{theorem.bipartite.metric.BPMA_is_O(logn)}
$\algbpma$ is $O(\log n)$-competitive for any metric and any concave delay function.
\end{theorem}

\subsection{Bipartite Matching on Weighted $\sigma$-HSTs}

Let $T = ( V(T), E(T), w)$ denote an arbitrary weighted $2$-HST defined with respect to a general metric $G = (V(G), E(G), w)$. We denote its height (i.e., the largest distance between $T$’s root and one of its leaves) by $h$. Recall that $h = O(\log n)$ where $n$ denotes the size of the original metric, $G$.

In order to define our algorithm in this case we make use of the DAG of counters, $F$, as defined in Subsection \ref{subsection.metric.matching_on_weighted_HSTs}. We alter $F$ to include 2 counters instead of each original counter - given a counter $z \in V(F)$ we define $z^+, z^-$ such that their capacities and slopes are defined as $y_{z^+} = y_{z^-} = y_z$ and $\alpha_{z^+} = \alpha_{z^-} = \alpha_z$. 

Even though we split every counter into two separate counters, we may sometimes refer to them as a single counter when clear from context. Thus, for example, when refering to the set of descendant counters of a given counter $z^+$ we in fact refer to all positive and negative counters that were created from a counter $z'$ that is a descendant of $z$ in $F$. 

Our algorithm, denoted by Bipartite-Metric-Algorithm ($\algbpma$) is defined as follows: requests will be associated with counters upon arrival (positive requests with positive counters and negative requests with negative counters). Then, over time, the requests will increase the levels of their counters. Once full, a single request from the counter will move to the next counter as defined by $F$ (i.e., the parent counter) and both the (former) positive and negative counters will be reset to 0. Furthermore, at any point in time if corresponding positive and negative counters both contain requests, match two of them (breaking ties arbitrarily). In order to define the rate at which we increase the counters we need the following definitions. Recall that given a counter $z\in V(F)$ we defined $F_z$ to be the set of all counters that are descendants of $z$ in $F$ (note that $z \in F_z$).

\begin{definition}
Given a counter $z^+$ or $z^-$ we define $F^b_z = \{\tilde{z}^+, \tilde{z}^-: \tilde{z}\in F_z\}$.
\end{definition}

\noindent To ease our notation we will henceforth denote $F^b_z$ simply as $F_z$.

\begin{definition}
Given counters $z^+$ and $z^-$ we denote by $P_z(t)$ (resp. $N_z(t)$) the number of positive (resp. negative) requests associated with counter $z^+$ (resp. $z^-$) at time $t$. Furthermore, we denote by $P_{F_z}(t)$ (resp. $N_{F_z}(t)$) the number of positive (resp. negative) requests associated with any counter in  $F^+_z$ (resp. $F^-_z$) at time $t$. Finally, we define the positive surplus of these counters as $sur_z(t) = P_{F_z}(t) - N_{F_z}(t)$.
\end{definition}

Now, to complete the definition of $\algbpma$, we define it such that a counter $z^+$ (resp. $z^-$) increases at time $t$ if and only if there exists a request associated with the counter and $sur_z(t) > 0$ (resp. $sur_z(t) < 0$). Furthermore, we increase it at a rate of $\alpha_z |sur_z(t)|$. Note that we increase all counters simultaneously. We formally define $\algbpma$ in Algorithm \ref{algorithm.bipartite.metric.BPMA} (deferred to the Appendix).

\begin{theorem}
\label{theorem.bipartite.metric.BPMA_is_O(1)_O(h)}
$\algbpma$ is $(O(1), O(h))$-competitive for any weighted 2-HST, $T$ with height $h$ and any concave delay function, $D$.
\end{theorem}

In order to prove our theorem we follow the proof presented for the single location case: we first upper bound the algorithm's cost by the overall increase in its counters and then upper bound this by the connection and delay costs of any (arbitrary) solution. 

\begin{definition}
Given a positive counter $z = z^+$ let $z'(t)$ denote the rate of increase of $z$ at time $t$. Hence, $z'(t) = \alpha_{z} |sur_k(t)|$ if there exists a positive request on $z$ and $sur_k(t) > 0$. Otherwise, $z'(t) = 0$. Define $z'(t)$ symmetrically for a negative counter, $z = z^-$.
\end{definition}

\begin{definition}
Given a request $r$ let $z_r'(t) = \alpha$ if $r$ causes counter $\hat{z}$ with slope $\alpha$ to increase at time $t$. Hence, $z_r'(t) = \alpha$ if $r$ is associated with counter $\hat{z}$ by $\algbpma$ at time $t$, $\hat{z}$ has slope $\alpha$ and $\hat{z}$ increases at time $t$. Otherwise, $z_r'(t) = 0$. 
\end{definition}

\subsubsection{Upper Bounding $\algbpma$'s Cost}

In this section we bound $\algbpma$'s cost (connection and delay) by the overall increase in its counters.

\begin{proposition}
\label{proposition.bipartite.metric.bound_algorithm_by_counters}
$\algbpma \leq 4 \sum_z \int_t z'(t)dt$.
\end{proposition}

We define $\delta_t(r), k_t(r), d_r'(t)$ and $s_r(t)$ as defined in Subsection \ref{subsection.metric.upper_bounding_MA_cost}. 
Let $\algbpma_d$ and $\algbpma_c$ denote the delay and connection costs of our algorithm. We first observe that due to the definition of our algorithm, its connection cost is upper bounded by the total increase in counters.

\begin{observation}
\label{observation.bipartite.metric.bound_bpma_connection_cost_by_increase_in_counters}
$\algbpma_c \leq \sum_z \int_t z'(t)dt$.
\end{observation}

\noindent In order to prove Proposition \ref{proposition.bipartite.metric.bound_algorithm_by_counters} we introduce the following lemmas.

\begin{lemma}
\label{lemma.bipartite.metric.real_delay_bounded_by_aux}
$\sum_r \int_t d_r'(t)dt \leq \sum_r  \sum_{k = 1}^{k_{m(r)}(r)-1} y_k  + \sum_r  \int_t s_r(t)dt$.
\end{lemma}

\begin{proof}
We argue that in fact, for any request $r$, $\int_t d_r'(t)dt \leq \sum_{k = 1}^{k_{m(r)}(r)-1} y_k  + \int_t s_r(t)dt$.

If $k_{m(r)}(r) > \delta_{m(r)}(r)$ then clearly $\int_t d_r'(t)dt \leq \sum_{k = 1}^{k_{m(r)}(r)-1} y_k$. Otherwise, assume $k_{m(r)}(r) \leq \delta_{m(r)}(r)$. The value $\int_t d_r'(t)dt$ constitutes of delay accumulated up to and including the $k_{m(r)}(r)-1$ linear piece (in the delay function) and delay accumulated during the rest of the time. Hence, 
\[
\int_t d_r'(t)dt = 
\sum_{k = 1}^{k_{m(r)}(r)-1} y_k + \int_{x_{k_{m(r)}(r)}}^{m(r) - a(r)} d_r'(t)dt,
\]
where $a(r)$ and $m(r)$ denote $r$'s arrival and matching times and $x_{k_{m(r)}(r)}$ denotes the time that the $k_{m(r)}(r)$'th linear piece begins in the delay function.

The value $\int_{x_{k_{m(r)}(r)}}^{m(r) - a(r)} d_r'(t)dt$ is upper bounded moment-wise by $\int_t s_r(t)dt$, since the delay accumulated by $s_r(t)dt$ is at least $\alpha_{k_{m(r)}(r)}$ and the delay accumulated by $d_r'(t)dt$ during that time is at most $\alpha_{k_{m(r)}(r)}$. Thus, $\int_t d_r'(t)dt \leq \sum_{k = 1}^{k_{m(r)}(r)-1} y_k  + \int_t s_r(t)dt$. Summing over all requests,
\begin{align}
\sum_r \int_t d_r'(t)dt \leq \sum_r  \sum_{k = 1}^{k_{m(r)}(r)-1} y_k  + \sum_r  \int_t s_r(t)dt.    
\end{align}
\end{proof}

\begin{lemma}
\label{lemma.bipartite.metric.bound_srt_by_increase_in_counters}
$\sum_r \int_t s_r(t)dt \leq 2 \sum_z \int_t z'(t)dt.$
\end{lemma}

\begin{proof}
To prove our lemma we will charge the value $\sum_r s_r(t)$ for a given time $t$ using a charging scheme $f_t$. Note that the scheme is defined for a given time $t$ and may change over time.

$f_t$ is defined as follows. We iterate over $F$'s counters containing requests in a buttom-up fashion; starting from counters farthest from the root (in terms of number of edges) while breaking ties arbitrarily. We begin such that all requests are unmarked and we will mark them as we iterate over the counters. 

Consider the iteration for which the counter $z$ was encountered. Let $A_z$ denote the requests belonging to counter $z$ (as we will see, we only mark requests from within $F_z$ and since we iterate over $F$ buttom-up, we are guaranteed that all requests in $A_z$ are unmarked). W.l.o.g. assume that they are positive. Let $B_{z}$ denote the set of unmarked negative requests associated with all counters in $F_z$ (note that $B_{z}$ might be empty). Consider an arbitrary subset $A \subset A_z$ of size $\min \{|A_z|, |B_z|\}$. Define $f_t$ on $A$ as a (arbitrary) 1-1 mapping to $B_z$. Mark all requests in $B_z$ and $A$.

The process partitions all the requests into three sets: the domain of $f_t$, denoted by $\mathcal{D}_t$, the image of $f_t$, denoted by $f_t(\mathcal{D}_t)$ and all the rest. We note that both $\mathcal{D}_t$ and $f_t(\mathcal{D}_t)$ may both contain positive and negative requests simultaneously.

By the definition of $f_t$, $f_t: \mathcal{D}_t \rightarrow f(\mathcal{D}_t)$ is 1-1 and always maps requests to requests with opposite polarity. Recall the definitions of $P_z(t) = P_z$ and $N_z(t) = N_z$. Note that since our algorithm matches requests of opposite polarity that are associated with the same counter, $\min \{P_z, N_z\} = 0$. Further, let $P_{F_z}$ and $N_{F_z}$ denote the number of positive and negative requests associated with any counter in $F_z$ (including $z$) - i.e., any counter in $F^+_z$ and $F^-_z$.  We say that $r \in z$ if $r$ is associated with $z$ at time $t$. 

Consider a counter $z$ with $z'(t) \neq 0$. Assume w.l.o.g. that the requests associated with $z$ are positive. Therefore, by the definition of our algorithm $sur_z(t) > 0$. Observe that by the definition of $f_t$, $|\{r: r\in z \land r  \in \mathcal{D}_t\}| = \min \{P_z, N_{F_z}\}$. By the definition of $P_z$, $|\{r: r\in z \land r  \in \mathcal{D}_t\}| + |\{r: r\in z \land r \not \in \mathcal{D}_t\}| = P_z$. Therefore, if $P_z \geq N_{F_z}$ then, 
\begin{align*}
|\{r: r\in z \land r \not \in \mathcal{D}_t\}| &= 
P_z - \min \{P_z, N_{F_z} \} = 
P_z - N_{F_z} \leq 
P_{F_z} - N_{F_z} = 
sur_z(t).    
\end{align*}
On the other hand, if $P_z \leq N_{F_z}$ then,
\begin{align*}
|\{r: r\in z \land r \not \in \mathcal{D}_t\}| = 
0 \leq 
sur_z(t).
\end{align*}
Therefore, in any case, by summing over all counters we are guaranteed that,

\begin{align}
\label{equation.bipartite.metric.bound_srt_by_increase_in_counters_1}
\sum_{r \not \in \mathcal{D}_t} s_r(t) &=
\sum_{r \not \in \mathcal{D}_t \land \forall z: r \not \in z} s_r(t) + \sum_{z} \sum_{r \not \in \mathcal{D}_t \land r \in z}s_r(t) \nonumber \\ &=
\sum_{z} \sum_{r \not \in \mathcal{D}_t \land r \in z}s_r(t)  =
\sum_{z} \alpha_z |\{r : r \not \in \mathcal{D}_t \land r \in z\}|  \leq
\sum_z \alpha_z |sur_z(t)|,
\end{align}
where the second equality is due to the fact that $s_r(t) = 0$ for requests that are not associated with any counter at time $t$, the second equality follows from the definition of $s_r(t)$ and the second inequality follows from our earlier discussion.

Observe that due to the fact that $f_t$ maps requests to requests that are their strict descendants in $F$, we are guaranteed that $s_r(t) \leq s_{f_t(r)}(t)$ for all $r \in \mathcal{D}_t$. Therefore, 
\begin{align}
\label{equation.bipartite.metric.bound_srt_by_increase_in_counters_2}
\sum_{r \in \mathcal{D}_t} s_r(t) \leq
\sum_{r \in \mathcal{D}_t} s_{f_t(r)}(t)  =
\sum_{r \in f_t(\mathcal{D}_t)} s_r(t) \leq
\sum_{r \not \in \mathcal{D}_t} s_r(t),
\end{align}
where the equality is due to the fact that $f_t$ is 1-1 and the second inequality is due to the fact that $\mathcal{D}_t \cap f_t(\mathcal{D}_t) = \emptyset$.

\noindent Combining the above and summing over all points in time,
\begin{align*}
\int_t \sum_r s_r(t)dt &= 
\int_t \sum_{r \in \mathcal{D}_t} s_r(t)dt + \int_t \sum_{r \not \in \mathcal{D}_t} s_r(t)dt \\&\leq
2 \int_t \sum_{r \not \in \mathcal{D}_t} s_r(t) dt \leq
2 \int_t \sum_z \alpha_z |sur_z(t)| dt =
2 \int_t \sum_z z'(t) dt.
\end{align*}
where the first inequality is due to equation (\ref{equation.bipartite.metric.bound_srt_by_increase_in_counters_2}) and the second is due to equation (\ref{equation.bipartite.metric.bound_srt_by_increase_in_counters_1}).
\end{proof}

\noindent We are now ready to prove Proposition \ref{proposition.bipartite.metric.bound_algorithm_by_counters}.

\begin{proof}[Proof of Proposition \ref{proposition.bipartite.metric.bound_algorithm_by_counters}]
We first observe that By the definition of $F$ and $k_{m(r)}(r)$, we are guaranteed that any request $r$ must have passed through delay counters $z_1, \ldots, z_{k_{m(r)}(r)-1}$. Due to the fact that a counter must be filled in order to be emptied, we are guaranteed that
\begin{align}
\label{equation.bipartite.metric.bound_BPSLA_by_counters_2}
\sum_r  \sum_{j = 1}^{k_{m(r)}(r)-1} y_j \leq \sum_k \int_t z'_k(t)dt,
\end{align}
since the capacity of delay counter $z_k$ is $y_k$. We note that Lemma \ref{lemma.single_location.real_delay_bounded_by_aux} clearly holds for $\algbpma$ as well. Therefore,
\begin{align*}
\algbpma &=
\algbpma_c + \algbpma_d \leq
\sum_z \int_t z'(t)dt + \algbpma_d \\&=
\sum_z \int_t z'(t)dt + \sum_r \int_t d_r'(t)dt \leq
\sum_z \int_t z'(t)dt + \sum_r  \sum_{k = 1}^{k_{m(r)}(r)-1} y_k  + \sum_r  \int_t s_r(t)dt \\&\leq 
2\sum_z \int_t z'(t)dt + \sum_r \int_t s_r(t)dt \leq 
4\sum_z \int_t z'(t)dt,
\end{align*}
where the first inequality follows from Observation \ref{observation.bipartite.metric.bound_bpma_connection_cost_by_increase_in_counters}, the second inequality follows from Lemma \ref{lemma.single_location.real_delay_bounded_by_aux}, the third inequality follows from equation (\ref{equation.bipartite.metric.bound_BPSLA_by_counters_2}) and the last inequality follows from Lemma \ref{lemma.bipartite.metric.bound_srt_by_increase_in_counters}.

\end{proof}

\subsubsection{Lower Bounding An Arbitrary Solution's Cost}

For an arbitrary solution to the given instance, we denote by $\algsol_c$ and $\algsol_d$ its connection and delay costs respectively. In this section we will lower bound $\algsol_c$ and $\algsol_d$ by the overall increase in counters by charging this increase to different matchings performed by $\algsol$. The following proposition states this formally.

\begin{proposition}
\label{proposition.bipartite.metric.bound_counters_by_OPT}
$\sum_z \int_t z'(t)dt \leq O(1)\cdot \algsol_c + O(h)\cdot \algsol_d.$
\end{proposition}

The rest of this section is dedicated towards the proof of Proposition \ref{proposition.bipartite.metric.bound_counters_by_OPT}. Throughout this section we denote by $F_z$ the set of all counters belonging to the subtree of $F^+$ rooted at $z^+$ and all counters belonging to the subtree of $F^-$ rooted at $z^-$ (note that this includes $z^+$ and $z^-$). We charge the increase in our algorithm's counters to $\algsol$ using the same flavor as in the proof of Proposition \ref{proposition.bipartite.single_location.bound_counters_by_opt}. We first introduce several definitions to aid us in our proof. We reuse the definition of $\delta^*_t(r)$ (see Definition \ref{definition.bipartite.bipartite.delta^*_t(r)}. 


\begin{definition}
\label{definition.bipartite.bipartite.delta^*_t(r)}
Given a request $r$ with arrival time $a(r)$ and some time $t$ such that $r$ is unmatched at time $t$ with respect to $\algsol$, define $\delta^*_t(r)$ such that $t - a(r) \in [x_{\delta^*_t(r)-1}, x_{\delta^*_t(r)})$.
\end{definition}

\begin{definition}
Let $o_z(t)$ denote the number of requests that (1) were given to a counter within $F_z$ upon arrival, (2) are unmatched at time $t$ with respect to $\algsol$ and (3) incur a momentary delay of at least $\alpha_z$ (i.e., $\alpha_{\delta^*_t(r)} \geq \alpha_z$). Furthermore, let $sur^*_z(t)$ denote the positive surplus taken over these requests.
\end{definition}

\begin{definition}
Given a counter $z$ define $k(z) \in \mathbbm{N}$ such that $\alpha_{k(z)}$ denotes the slope of $z$.
\end{definition}

\begin{definition}
For a given time interval $I = [a,b)$, let $\mathcal{E}^z_{I}$ denote the number of requests $r$ that were given to a counter in $F_z$ upon arrival, such that either (1) $\delta^*_t(r)$ changed from $k(z)$ to $k(z)+1$ at time $t \in I$, (2) $r$ is matched (upon arrival) at time $t\in I$ by $\algsol$ to a request $r'$ such that $\delta^*_t(r') \geq k(z)$ or (3) $\algsol$ matched the request at time $t \in I$ through an edge that is an ancestor of $z$ in $F$.
\end{definition}

As in \cite{Min_Cost_Bipartite_Perfect_Matching_with_Delays} we define a potential function to aid in our proof. Specifically, let $\phi_z(t) = 2 y_z |sur_z(t) - sur^*_z(t)|$. Furthermore, we partition our time axis into intervals as in Definition \ref{definition.single_location.counter_intervals}.


\begin{lemma}
\label{lemma.bipartite.metric.yossis_lemma}
Let $I^z_i$ denote some phase interval $I^z_i = [t^z_i, t^z_{i+1})$ defined with respect to counter $z$. Then, $\int_{t \in I_i}z'(t)dt + \Delta_i(\phi_z (t)) \leq 4\alpha_z \cdot \int_{t \in I^z_i} o_z(t) dt + 2\mathcal{E}^z_{I^z_i}y_z$.
\end{lemma}

\begin{proof}
Let $\mathcal{E} = \mathcal{E}^z_{I^z_i}$. Observe that $f_z(t) = sur^*_z(t) - sur_z(t)$ can change if and only if: (1) $\algbpma$ moves a request from $z$ to its parent in $F$, (2) $\delta^*_t(r)$ changed from $k(z)$ to $k(z)+1$ during $t \in I_i^z$, (3) $r$ is matched (upon arrival) at time $t\in I_i^z$ by $\algsol$ to a request $r'$ such that $\delta^*_t(r') \geq k(z)$ or (4) $\algsol$ matched the request at time $t \in I^z_i$ through an edge that is an ancestor of $z$ in $F$. 

From here on out, the proof continues identically to the proof of Lemma \ref{lemma.bipartite.single_location.yossis_lemma}. We therefore defer the full proof to the Appendix.

\end{proof}

\noindent We are now ready to prove Proposition \ref{proposition.metric.bound_counters_by_OPT}.

\begin{proof}[Proof of Proposition \ref{proposition.bipartite.metric.bound_counters_by_OPT}]
By Lemma \ref{lemma.bipartite.metric.yossis_lemma} and due to the fact that $\phi_z = 0$ at the beginning and end of the instance, by summing over all phases of a counter $z$ we get that,
\begin{align}
\int_t z'(t)dt \leq 4\int_t \alpha_z o_z(t)dt + 2y_z \mathcal{E}_z,
\end{align}
where $\mathcal{E}_z = \sum_i \mathcal{E}^z_{I_i}$. We first split the set of all counters to the set of delay counters, $Z_d$ and the set of edge counters, $Z_e$. We upper bound each separately, beginning with $Z_d$. Throughout this section, given a counter $z$ and request $r$ we will denote by $r \in F_z$ the case that $r$ was given to a counter from $F_z$ upon arrival.

We upper bound $\sum_{Z_d}\int_t \alpha_z o_z(t)dt$ and $\sum_{Z_d}y_z \mathcal{E}_z$ separately, beginning with the former. By the definition of $o_z(t)$ we get,
\begin{align}
\label{equation.bipartite.metric.bound_counters_by_OPT_1}
\sum_{Z_d} \int_t \alpha_z o_z(t)dt =&
\sum_{Z_d} \int_t \alpha_z \sum_r \mathbbm{1}\{ r \in F_z \land \delta^*_t(r) \leq k(z)\} \nonumber \\=&
\sum_r \int_t \sum_{Z_d} \alpha_z \mathbbm{1}\{ r \in F_z \land \delta^*_t(r) \leq k(z)\}.
\end{align}

Note that given request $r$, there is at most a single delay counter of level $k$, $z_k$, such that $r \in F_z$, due to Observation \ref{observation.metric.properties_of_F}. Therefore, due to the fact that the delay slopes decrease exponentially,
\begin{align}
\label{equation.bipartite.metric.bound_counters_by_OPT_2}
\sum_{Z_d} \alpha_z \mathbbm{1}\{ r \in F_z \land \delta^*_t(r) \leq k(z)\} \leq 
\sum_{k \geq \delta^*_t(r)} \alpha_k \leq 
2\alpha_{\delta^*_t(r)}.
\end{align}

\noindent Combining equations (\ref{equation.bipartite.metric.bound_counters_by_OPT_1}) and (\ref{equation.bipartite.metric.bound_counters_by_OPT_2}) yields,

\begin{align}
\label{equation.bipartite.metric.bound_counters_by_OPT_3}
\sum_{Z_d} \int_t \alpha_z o_z(t)dt \leq 
\sum_r \int_t 2\alpha_{\delta^*_t(r)} =
2 \cdot \algsol_d.
\end{align}

\noindent Next we upper bound $\sum_{Z_d}y_z \mathcal{E}_z$. Recall the definition of $\rho(r)$ (as defined in the proof of Proposition \ref{proposition.single_location.bound_counters_by_OPT}). By the definition of $\mathcal{E}_z$ for delay counters we are guaranteed that a request is counted towards $\mathcal{E}_z$ if and only if $r \in F_z \land \Big(\rho(r) \geq \sum_{i=1}^{k(z)} y_i \lor \kappa(r) \geq \sum_{i=1}^{k(z)}y_i\Big)$. Therefore,
\begin{align}
\label{equation.bipartite.metric.bound_counters_by_OPT_4}
\sum_{Z_d}y_z \mathcal{E}_z &=
\sum_{Z_d}y_z \sum_r \mathbbm{1}\Big\{ r \in F_z \land \big(\rho(r) \geq \sum_{i=1}^{k(z)} y_i \lor \kappa(r) \geq \sum_{i=1}^{k(z)}y_i\big)\Big\}.
\end{align}

Again, due to the fact that given a request $r$ there is at most a single delay counter of level $k$, $z_k$ such that $r\in F_z$, we get,

\begin{align}
\label{equation.bipartite.metric.bound_counters_by_OPT_5}
\sum_{Z_d} y_z \mathbbm{1}\Big\{ r \in F_z \land \big(\rho(r) \geq \sum_{i=1}^{k(z)} y_i \lor \kappa(r) \geq \sum_{i=1}^{k(z)}y_i\big)\Big\} \leq \rho(r) + \kappa(r).
\end{align}

\noindent Combining equations (\ref{equation.bipartite.metric.bound_counters_by_OPT_4}) and (\ref{equation.bipartite.metric.bound_counters_by_OPT_5}) yields,

\begin{align}
\label{equation.bipartite.metric.bound_counters_by_OPT_6}
\sum_{Z_d}y_z \mathcal{E}_z \leq \sum_r (\rho(r) + \kappa(r)) \leq 2 \cdot \algsol_d + 2 \cdot \algsol_c.
\end{align}

We turn to consider $Z_e$. We first upper bound $\sum_{Z_e}\int_t \alpha_z o_z(t)dt$. Equation (\ref{equation.bipartite.metric.bound_counters_by_OPT_1}) holds for $Z_e$ as well (replacing $Z_d$). Due to Observation \ref{observation.metric.properties_of_F} we are guaranteed that there are at most $h$ (the height of the HST) counters that satisfy $r \in F_z$ for a given request $r$. Therefore,
\begin{align*}
\sum_{Z_e} \alpha_z \mathbbm{1}\{ r \in F_z \land \delta^*_t(r) \leq k(z)\} &\leq 
\alpha_{\delta^*_t(r)} \sum_{Z_e} \mathbbm{1}\{ r \in F_z \land \delta^*_t(r) \leq k(z)\} \\&\leq
\alpha_{\delta^*_t(r)}h.
\end{align*}

\noindent Therefore, combining the above with equation (\ref{equation.bipartite.metric.bound_counters_by_OPT_1}) (with $Z_e$ in lieu of $Z_d$) yields,
\begin{align}
\label{equation.bipartite.metric.bound_counters_by_OPT_10}
\sum_{Z_e} \int_t \alpha_z o_z(t)dt \leq h \sum_r \int_t \alpha_{\delta^*_t(r)} = h \cdot \algsol_d.
\end{align}

Next, we upper bound $\sum_{Z_e}y_z \mathcal{E}_z$. By the definition of $\mathcal{E}_z$ for edge counters we are guaranteed that a request is counted towards $\mathcal{E}_z$ if and only if $r \in F_z \land \Big(\rho(r) \geq \sum_{i=1}^{k(z)} y_i \lor \kappa(r) \geq y_z\Big)$. Also note that due to the fact that $z_e$'s slope is defined as the first delay counter's slope encountered on the way to the root, $\sum_{i=1}^{k(z)} y_i \geq y_{z}$. Therefore,

\begin{align}
\label{equation.bipartite.metric.bound_counters_by_OPT_7}
\sum_{z \in Z_e} y_z \mathcal{E}_z &=
\sum_{z \in Z_e} y_z \sum_r \mathbbm{1}\Big \{r \in F_z \land \Big(\rho(r) \geq \sum_{i=1}^{k(z)} y_i \lor \kappa(r) \geq y_z\Big)\Big \} \nonumber \\&\leq 
\sum_r \sum_{z \in Z_e} y_z \mathbbm{1}\Big \{r \in F_z \land \Big(\rho(r) \geq y_z \lor \kappa(r) \geq y_z\Big)\Big \}.
\end{align}

\noindent Due to the fact that the edge counters' capacities decrease exponentially, we are guaranteed that,

\begin{align}
\label{equation.bipartite.metric.bound_counters_by_OPT_8}
\sum_{z \in Z_e} y_z \mathbbm{1}\Big \{r \in F_z \land \Big(\rho(r) \geq y_z \lor \kappa(r) \geq y_z\Big)\Big \} \leq
2\rho(r) + 2\kappa(r).
\end{align}

\noindent Combining equations (\ref{equation.bipartite.metric.bound_counters_by_OPT_7}) and (\ref{equation.bipartite.metric.bound_counters_by_OPT_8}) yields,
\begin{align}
\label{equation.bipartite.metric.bound_counters_by_OPT_9}
\sum_{z \in Z_e} y_z \mathcal{E}_z \leq 
\sum_r 2\rho(r) + 2\kappa(r) \leq 
4 \cdot \algsol_d + 4 \cdot \algsol_c.
\end{align}

\noindent Combining the above, we get,
\begin{align*}
\sum_z \int_t z'(t)dt &\leq 
\sum_z (4\int_t \alpha_z o_z(t)dt + 2y_z \mathcal{E}_z) \\&=
\sum_{z \in Z_d} (4\int_t \alpha_z o_z(t)dt + 2y_z \mathcal{E}_z) + \sum_{z \in Z_e} (4\int_t \alpha_z o_z(t)dt + 2y_z \mathcal{E}_z) \\&\leq
O(h) \cdot \algsol_d + O(1) \cdot \algsol_c,
\end{align*}
where the first inequality follows from equation (\ref{equation.bipartite.metric.bound_counters_by_OPT_1}) and the second inequality follows from equations (\ref{equation.bipartite.metric.bound_counters_by_OPT_3}), (\ref{equation.bipartite.metric.bound_counters_by_OPT_6}), (\ref{equation.bipartite.metric.bound_counters_by_OPT_10}) and (\ref{equation.bipartite.metric.bound_counters_by_OPT_9}).
\end{proof}

\noindent We are now ready to prove Theorems \ref{theorem.bipartite.metric.BPMA_is_O(1)_O(h)} and \ref{theorem.bipartite.metric.BPMA_is_O(logn)}

\begin{proof}[Proof of Theorem \ref{theorem.bipartite.metric.BPMA_is_O(1)_O(h)}]
Follows from Propositions \ref{proposition.bipartite.metric.bound_algorithm_by_counters} and \ref{proposition.bipartite.metric.bound_counters_by_OPT}.
\end{proof}

\begin{proof}[Proof of Theorem \ref{theorem.bipartite.metric.BPMA_is_O(logn)}]
Follows from Theorem \ref{theorem.bipartite.metric.BPMA_is_O(1)_O(h)} and Lemma \ref{lemma.transfer_competitiveness_from_hst_to_general_metric} (note that the lemma clearly holds for the bichromatic case as well).
\end{proof}

\section{Concluding Remarks and Open Problems}
\label{section.conclusion}
In this paper we study the online problem of {\em minimum-cost perfect matching with concave delays} in two settings: the monochromatic setting and the bichromatic setting.
For each setting, we first present an $O(1)$-competitive deterministic online algorithm for the single location cases and then generalize our ideas in order to design $O(\log n)$-competitive randomized algorithms for the metric cases. 

The problem of online minimum-cost perfect matching with linear delays has been extensively studied in the deterministic setting as well. To the best of our knowledge the results in the linear case do not convey as is to the case that the delays are concave. An interesting avenue to pursue, is therefore to try and generalize the ideas presented in this paper in order to handle the problem in the deterministic setting.

\bibliography{bib}
\bibliographystyle{plain}
\appendix

\section{Deferred Proofs from Section \ref{section-notation}}
\begin{proof}[Proof of Lemma \ref{lemma.approximate_concave_with_piece_wise_linear}]

Given the original concave delay function $D(\cdot)$, $f(\cdot)$ is constructed in an iterative method. We define each piece in our piece-wise linear function iteratively. The first piece begins at $(0,0)$ (i.e., $f(0) = 0$). The first piece's slope, denoted by $\alpha_1$, is defined as $D'(0) / 2$. We consider two cases: either such $f$ does not intersect $D$ at any point $x > 0$ or it does. In the former case, we simply set $f$ as a linear function (i.e., $f(0) = 0$ with slope $\alpha_1 = D'(0) / 2$). In the latter case, let $x_1$ denote the point at which $f(x_1) = D(x_1)$. In this case we continue to define our next linear piece as defined next.

Assume $f$ is defined for the $i-1$ first linear pieces. We define the $i$'th linear piece. By our definition $f(x_{i-1}) = D(x_{i-1})$. Set $\alpha_i$ as $D'(x_{i-1})/2$ and consider two cases. Either $f$ intersect $D$ at some $x > x_{i-1}$ or it does not. If it does, let $x_i$ denote that point, set $f(x_i) = D(x_i)$ and continue to the next linear piece. If it does not, let $f$ continue with that slope to $x = \infty$.

Since $D$ is concave, $D'$ is non-increasing. Furthermore, the linear function between $x_{i-1}$ and $x_i$ must be below $D$ (throughout that interval) and thus, $\alpha_{i-1} \leq D'(x_i)$. Therefore, $\alpha_{i-1} \leq D'(x_i) = 2\alpha_i$ meaning that $f$'s slopes are exponentially decreasing.

We are left to show that $f(x) \leq D(x) \leq 2f(x)$. Clearly, by the definition of $f$, $f(x) \leq D(x)$ for all $x$. On the other hand, again due to the fact that $D'(x)$ is non-increasing, given any $x \in [x_{i-1}, x_i]$, we have $D(x) \leq D(x_{i-1}) + D'(x_{i-1})\cdot (x - x_{i-1})$. Since $f(x) = D(x_{i-1}) + \frac{D'(x_{i-1})}{2} \cdot (x - x_{i-1})$, we have
\begin{align*}
D(x) - f(x) &\leq 
D'(x_{i-1})\cdot (x - x_{i-1}) - \frac{D'(x_{i-1})}{2} \cdot (x - x_{i-1}) \\&= 
\frac{D'(x_{i-1})}{2} \cdot (x - x_{i-1}) \leq 
f(x),    
\end{align*}
which yields $D(x) \leq 2 f(x)$.
\end{proof}

\section{Deferred Algorithms from Section \ref{section.single_location}}

\begin{algorithm}[H]
\SetAlgoNoLine
\KwIn{a set of requests arriving online; counters $\{z_k\}_{k}$ for every linear piece in the delay function, each $z_k$ with a slope $\alpha_k$ and a capacity $y_k$.}
\textbf{Upon} arrival of request $r$ \textbf{do} \\
\Indp			
	Add $r$ to $z_1$\;
\Indm
\For{all time $t$}
{
\While{there exists a pending request $r$}{
    \If{there exist 2 requests associated with $z_k$}
    {
        Match them and remove them from $z_k$\;
    }
    \If{$z_k$ has a request associated with it AND an even number of requests are associated with counters $z_1, \dots, z_{k-1}$}
    {
    Increase $z_k$ continuously over time at a rate of $\alpha_k$\;
    \If{$z_k$ reaches its capacity $y_k$}
    {
        Move $r$ from $z_k$ to $z_{k+1}$\;
        Reset $z_k$ to 0.
    }
    }
}
}
\caption{Single-Location-Algorithm ($\algsla$)}
\label{algorithm.single_location}
\end{algorithm}

\section{Deferred Algorithms and Proofs from Section \ref{section.metric}}

\begin{algorithm}[H]
\SetAlgoNoLine
\KwIn{a set of requests arriving online; an HST $T$; a DAG $F(T, D)$, of counters.}
\textbf{Upon} arrival of request $r$ at $v \in V(T)$ \textbf{do} \\
\Indp
    Let $e$ denote $v$'s single edge in $T$\;
    Let $z$ denote the (single) leaf encountered while moving downward from $z_e$ in $F$\;
    Associate $r$ with $z$\;
\Indm
\For{all time $t$}
{
\While{there exists a pending request $r$}{
    \If{there exist 2 requests associated with $z$}
    {
        Match them and remove them from $z$\;
    }
    \If{$z$ has a request associated with it AND an odd number of requests are associated with $F_z$}
    {
        Increase $z$ continuously over time at a rate of $\alpha_z$\;
    \If{$z$ reaches its capacity $y_z$}
    {
        Move $r$ from $z$ to its parent in $F$\;
        Reset $z$ to 0.
    }
    }
}    
}
\caption{Metric-Algorithm ($\algma$)}
\label{algorithm.metric}
\end{algorithm}

\begin{proof}[Proof of Lemma \ref{lemma.metric.interval_is_odd}]
Denote $I^{\hat{z}}_i$ simply as $I_i$ and let $t \in I_i$ denote a time before any request from $R(I_i)$ arrived. If $i = 0$ then clearly the number of requests belonging to $F_{\hat{z}}$ is even (in fact it is 0). Otherwise, at time $t_i$ a request belonging to $\hat{z}$ caused $\hat{z}$ to increase and then zero. Furthermore, this request moved to a counter $\not \in F_{\hat{z}}$. By the definition of our algorithm (yielding the fact that requests may leave $F_{\hat{z}}$ only if they are matched within $F_{\hat{z}}$) and due to the fact that $\hat{z}$ increased and then the request left $F_{\hat{z}}$, we are guaranteed that at time $t$ there are an even number of requests belonging to $F_{\hat{z}}$.

Let $t' \in I_i$ denote some time at which $\hat{z}$ increases. Between $t$ and $t'$ only requests from $R(I_i)$ may arrive and be given to counters within $F_{\hat{z}}$. Furthermore, no request may leave $F_{\hat{z}}$ other than if it is matched to another request from $F_{\hat{z}}$. Therefore, due to the fact that $\hat{z}$ increased at time $t'$ we are guaranteed that an odd number of requests from $R(I_i)$ have arrived and thus $t' \in I_i^{odd}$. 
\end{proof}

\begin{proof}[Proof of Lemma \ref{lemma.metric.last_interval_is_even}]
Denote $I^{\hat{z}}_{m_{\hat{z}}-1} = [t^{\hat{z}}_{m_{\hat{z}}-1}, t^{\hat{z}}_{m_{\hat{z}}})$ simply as $I_{m-1} = [t_{m-1}, t_{m})$. At time $t_{m-1}$, the number of requests within $F_{\hat{z}}$ must be even since $\hat{z}$ increased and the request from $\hat{z}$ moved to a counter $\not \in F_{\hat{z}}$. Since $t_m = \infty$ we are guaranteed that no request from $F_{\hat{z}}$ will move to a counter $\not \in F_{\hat{z}}$ during $[t_{m-1}, t_m)$. Finally, by Remark \ref{remark.MA_terminates} we are guaranteed that all requests in $F_{\hat{z}}$ match during $[t_{m-1}, t_m = \infty)$. Since the algorithm is defined such that it only matches requests on the same counter, we are guaranteed that $|R(I_{m-1}, \hat{z})|$ is even.
\end{proof}

\begin{proof}[Proof of Lemma \ref{lemma.metric.parity_of_last_interval_OPT}]
Denote $I^{\hat{z}}_{m_{\hat{z}}-1}$ simply as $I_{m-1}$ and $I_{m_{\hat{z}}-1, {\hat{z}}}^{odd}$ simply as $I_{m-1}^{odd}$. We assume the first condition fails. Let $t \in I_{m-1}^{odd}$. By Lemma \ref{lemma.metric.last_interval_is_even} we are guaranteed that the number of request that will arrive after after $t$ and will be given to $F_{\hat{z}}$ upon arrival is odd - denote this set of requests as $\mathcal{T}$. Due to the fact that the first condition failed, we are guaranteed that there exists a request in $\mathcal{T}$ that is matched to a request that does not belong to $\mathcal{T}$ and that was given to $F_{\hat{z}}$ upon arrival. Since $I_{m-1}$ is the last interval, this request must be unmatched at time $t$ guaranteeing that $\optim$ is live with respect to $I_{m-1}$ and $\hat{z}$ at time $t$.
\end{proof}

\begin{proof}[Proof of Lemma \ref{lemma.metric.middle_interval_is_odd}]
Denote $I^{\hat{z}}_i$ simply as $I_i$. Let $a \in I_i = [t_i, t_{i+1})$ denote some time before any request from $R(I_i, \hat{z})$ arrived and let $b \in [t_i, t_{i+1})$ denote some time after all requests from $R(I_i, \hat{z})$ have arrived. At time $t_i$ a request moved up from $\hat{z}$, therefore immediately after the request moved, $F_{\hat{z}}$ contained an even number of requests. Thus, this holds for $a$ as well since no requests arrived and no request left $F_{\hat{z}}$ (other than if it was matched to another request from $F_{\hat{z}}$). On the other hand, at time $b$ there are an odd number of requests in $F_{\hat{z}}$ since $\hat{z}$ increases. Due to the fact that no request can leave $F_{\hat{z}}$ during $[a,b]$ unless it is matched to another such request, we are guaranteed that $|R(I_i, \hat{z})|$ is odd.
\end{proof}

\begin{proof}[Proof of Lemma \ref{lemma.metric.middle_interval_parity_OPT}]
We assume that the first condition does not hold. Thus, we may assume that (1) all requests from $R(I^{\hat{z}}_i \cup I^{\hat{z}}_{i+1})$ were matched to requests that were given to $F_{\hat{z}}$ upon arrival. Further, we assume that the second condition does not hold for $I^{\hat{z}} = I^{\hat{z}}_i$. Therefore, there exists some time $t' \in I^{\hat{z}}_i$ such that (2) an odd number of requests from $R(I^{\hat{z}}_i)$ have arrived, (3) the requests from $R(I^{\hat{z}}_i)$ that arrived until time $t'$ have all been matched (by $\optim$'s matching) and (4) no request that arrived prior to $t'$ and that was given to $F_{\hat{z}}$ will be matched to $R(I^{\hat{z}}_i \cup I^{\hat{z}}_{i+1})$ in the future.

Let $t\in I^{\hat{z}}_{i+1}$ denote some time such that an odd number of requests from $R(I^{\hat{z}}_{i+1})$ have arrived. If at time $t$ there is an unmatched request (with respect to $\optim$'s matching) from $R(I^{\hat{z}}_i \cup I^{\hat{z}}_{i+1})$ then the lemma holds. Therefore, we assume the contrary. By (1), (2) and Lemma \ref{lemma.metric.middle_interval_is_odd}, we are guaranteed that the number of requests that arrived between $t'$ and $t$ and were given to $F_{\hat{z}}$ upon arrival, is odd. By our assumption on $t$, we are guaranteed that these requests are all matched at time $t$. Therefore, one of these requests must be matched to a request that arrived prior to $t'$ and that was given to $F_{\hat{z}}$ upon arrival - which contradicts (4), proving our lemma.
\end{proof}

\section{Deferred Algorithms from Section \ref{section.bipartite.single_location}}

\begin{algorithm}[H]
\SetAlgoNoLine
\KwIn{a set of requests arriving online; counters $\{z_k^+, z_k^-\}_{k}$ for every linear piece in the delay function, each $z_k^+$ and $z_k^-$ with a slope $\alpha_k$ and a capacity $y_k$.}
\textbf{Upon} arrival of request $r$ \textbf{do} \\
\Indp
    \textbf{if} $r$ is positive \textbf{then} add $r$ to $z_1^+$ \textbf{else} add $r$ to $z_1^-$\;
\Indm
\For{all time $t$}
{
\While{there exists a pending request $r$}{
    \If{there exist requests $r$ and $r'$ associated with $z_k^+$ and $z_k^-$}
    {
        Match $r$ and $r'$\;
        Remove $r$ and $r'$ from the counters\;
    }
    \If{If there exists a request associated with $z_k^+$ and $sur_k(t) > 0$}
    {
        Increase $z_k^+$ continuously at a rate of $\alpha_k|sur_k(t)|$\;
        \If{$z_k^+$ reaches its capacity $y_k$}
        {
            Move a request associated with $z_k^+$ to $z_{k+1}^+$\;
            Reset both $z_k^+$ and $z_k^-$ to 0\;
        }
    }
    \If{If there exists a request associated with $z_k^-$ and $sur_k(t) < 0$}
    {
        Increase $z_k^-$ at a rate of $\alpha_k|sur_k(t)|$\;
        \If{$z_k^-$ reaches its capacity $y_k$}
        {
            Move a request associated with $z_k^-$ to $z_{k+1}^-$\;
            Reset both $z_k^+$ and $z_k^-$ to 0.
        }
    }
}
}
\caption{Bipartite-Single-Location-Algorithm ($\algbpsla$)}
\label{algorithm.bipartite.single_location.BPSLA}
\end{algorithm}

\begin{remark}
Note that all counters $z_k$, are considered simultaneously and that all tie breaking in Algorithm \ref{algorithm.bipartite.single_location.BPSLA} is done arbitrarily (specifically, lines 5, 10 and 15).
\end{remark}

\begin{proof}[Proof of Lemma \ref{lemma.single_location.interval_parity_of_SLA}]
Denote $I_{i}^k$ simply as $I_i$ and let $t \in [t_i, t_{i+1})$ denote a time before any request from $R(I_i)$ arrived. If $i=0$ then clearly the number of requests associated with $z_1, \ldots, z_k$ is even (in fact it is 0). Otherwise, at time $t_i$ a request associated with $z_k$ caused $z_k$ to increase and then zero. Furthermore, this request moved to counter $z_{k+1}$. By the definition of our algorithm and due to the fact that $z_k$ increased and then the request left $z_k$, we are guaranteed that at time $t$ there are an even number of requests associated with $z_1, \ldots, z_k$. 

Let $t' \in I_i$ denote some time for which $z_k$ increases. Between $t$ and $t'$ only requests from $R(I_i)$ may arrive. Furthermore, no request may move from $z_k$ to $z_{k+1}$ (since otherwise the interval would have ended) and any requests to leave $z_1, \ldots, z_k$ must do so in pairs. Therefore, due to the fact that $z_k$ increased at time $t'$ we are guaranteed that an odd number of requests from $R(I_i)$ have arrived and thus $t' \in I_i^{odd}$.
\end{proof}

\begin{proof}[Proof of Lemma \ref{lemma.single_location.last_interval_is_even}]
Denote $I^k_{m_k-1} = [t^k_{m_k-1}, t^k_{m_k})$ simply by $I_{m-1} = [t_{m-1}, t_{m})$. At time $t_{m-1}$ when an unmatched request is moved from $z_k$ to $z_{k+1}$, the total number of requests associated with $z_{k+1}, \ldots, z_d$ must be even. 
This is because no unmatched request is moved from $z_k$ to $z_{k+1}$ after time $t_{m-1}$ and all the requests associated with $z_{k+1}, \ldots, z_d$ are formed into pairs (according to $\algsla$).
Besides, at time $t_{m-1}$, the number of requests associated with $z_1, \ldots, z_{k-1}$ must also be even.
This is because $z_k$ is increasing at this moment. 
According to $\algsla$, no request in $z_j$ ($j < k$) is unmatched.
Therefore, the number of requests already arrived at time $t_{m-1}$ (i.e., the requests besides $R(I_{m-1})$) must be even. 
Since $t_m = \infty$, this guarantees that $|R(I_{m-1})|$ is also even.
\end{proof}

\begin{proof}[Proof of Lemma \ref{lemma.single_location.last_interval_parity_of_OPT}]
Denote $I_{m_k-1}^k$ simply as $I_{m-1}$ and $I_{m_k-1,k}^{odd}$ simply as $I_{m-1}^{odd}$. Let $t \in I^{odd}_{m-1}$ denote a time such that an odd number of requests from $R(I_{m-1})$ arrived. Due to the fact that the defined time intervals partition our entire instance and that $I_{m-1}$ is the last interval, we are guaranteed that the number of requests to arrive prior to $t$ is odd and thus there must be an unmatched request with respect to $\optim$'s matching, at time $t$, that will be matched to $R(I_{m-1})$. Therefore, $\optim$ is live with respect to $I_{m-1}$ at time $t$.
\end{proof}

\begin{proof}[Proof of Lemma \ref{lemma.single_location.interval_is_odd}]
Denote $I_i^k$ simply as $I_i$. Let $a \in [t_i, t_{i+1})$ denote some time before any request from $R(I_i)$ has arrived and let $b \in [t_i, t_{i+1})$ denote some time after all requests from $R(I_i)$ have arrived. At time $a$ and at time $b$ the parity of requests in $z_{k+1}, \ldots, z_d$ must be the same (since we are guaranteed by the definition of $\{t_i\}_i$ that no request moved during this time to $z_{k+1}$). Furthermore, we are guaranteed that the number of requests associated with $z_1, \ldots, z_{k-1}$ at both times are even (since both at time $t_i$ and at time $b$, counter $z_k$ increased and no request arrived during $(t_i, a)$). Finally, at time $a$, $z_k$ does not contain a request while at time $b$ it does and all matched requests must be even. Thus, overall, the parity of the number of overall requests to arrive at each point is different resulting in the fact that $|R(I_i)|$ is odd.
\end{proof}

\begin{proof}[Proof of Lemma \ref{lemma.single_location.middle_interval_parity_OPT}]
We denote $I_i^k$ and $I_{i+1}^k$ simply as $I_i$ and $I_{i+1}$. We assume that the condition does not hold for $I_i$ (since otherwise the lemma would be proven). Therefore, there exists some time $t' \in I_i$ such that (1) an odd number of requests from $R(I_i)$ have arrived, (2) the requests from $R(I_i)$ that arrived until time $t'$ have all been matched by $\optim$'s matching and (3) no request that arrived prior to $t'$ will be matched to $R(I_i \cup I_{i+1})$ in the future.

Let $t \in I_{i+1}$ denote some time such that an odd number of requests from $R(I_{i+1})$ arrived. If at time $t$ there is an unmatched request (with respect to $\optim$'s matching) from $R(I_i \cup I_{i+1})$ then the lemma holds. Therefore, we assume the contrary. By (1) and Lemma \ref{lemma.single_location.interval_is_odd} the number of requests that arrived between $t'$ and $t$ is even. Therefore, one of the requests that arrived between $t'$ and $t$ would have had to have been matched to a request that arrived before $t'$. However, this leads to a contradiction due to (2) and (3).
\end{proof}

\section{Deferred Algorithms from Section \ref{section.bipartite.metric}}

\begin{algorithm}[H]
\SetAlgoNoLine
\KwIn{a set of requests arriving online; an HST $T$; a DAG $F(T, D)$, of counters.}
\textbf{Upon} arrival of request $r$ at $v \in V(T)$ \textbf{do} \\
\Indp
    Let $e$ denote $v$'s single edge in $T$\;
    Let $z$ denote the (single) leaf encountered while moving downward from $z_e$ in $F$\;
    \textbf{if} $r$ is positive \textbf{then} add $r$ to $z^+$ \textbf{else} add $r$ to $z^-$\;
\Indm
\For{all time $t$}
{
\While{there exists a pending request $r$}{
    \If{there exist requests $r$ and $r'$ associated with $z^+$ and $z^-$}
    {
        Match $r$ and $r'$\;
        Remove $r$ and $r'$ from the counters\;
    }
    \If{If there exists a request associated with $z^+$ and $sur_z(t) > 0$}
    {
        Increase $z^+$ continuously over time at a rate of $\alpha_z|sur_z(t)|$\;
        \If{$z^+$ reaches its capacity $y_z$}
        {
            Move a request associated with $z^+$ to its parent in $F$\;
            Reset both $z^+$ and $z^-$ to 0\;
        }
    }
    \If{If there exists a request associated with $z^-$ and $sur_z(t) < 0$}
    {
        Increase $z^-$ continuously over time at a rate of $\alpha_z|sur_z(t)|$\;
        \If{$z^-$ reaches its capacity $y_z$}
        {
            Move a request associated with $z^-$ to its parent in $F$\;
            Reset both $z^+$ and $z^-$ to 0.
        }
    }
}
}
\caption{Bipartite-Metric-Algorithm ($\algbpma$)}
\label{algorithm.bipartite.metric.BPMA}
\end{algorithm}

\begin{remark}
Note that all counters $z$, are considered simultaneously and that all tie breaking in Algorithm \ref{algorithm.bipartite.metric.BPMA} is done arbitrarily (specifically, lines 5, 10 and 15).
\end{remark}

\begin{proof}[Proof of Lemma \ref{lemma.bipartite.metric.yossis_lemma}]
Let $\mathcal{E} = \mathcal{E}^z_{I_i}$. Observe that $f_z(t) = sur^*_z(t) - sur_z(t)$ can change if and only if: (1) $\algbpma$ moves a request from $z$ to its parent in $F$, (2) $\delta^*_t(r)$ changed from $k(z)$ to $k(z)+1$ $t \in I$, (3) $r$ is matched (upon arrival) at time $t\in I$ by $\algsol$ to a request $r'$ such that $\delta^*_t(r') \geq k(z)$ or (4) $\algsol$ matched the request at time $t \in I$ through an edge that is an ancestor of $z$ in $F$.

A change of type (1) can only happen once during the phase and changes of types (2), (3) and (4) happen exactly $\mathcal{E}$ times. Therefore, $|\Delta_i(f_z(t))| \leq \mathcal{E}+1$. Note that always, $\Delta_i|f_z(t)| \leq |\Delta_i(f_z(t))|$ and thus overall,

\begin{align}
\label{equation.bipartite.metric.yossis_lemma_1}
\Delta_i|f_z(t)| \leq |\Delta_i(f_z(t))| \leq \mathcal{E}+1.
\end{align}

We first consider the case that the interval is not last (i.e., $I_i  = [t_i, t_{i+1}) \neq [t_{m-1}, t_m = \infty)$). Assume w.l.o.g. that the phase ends once the algorithm moves a positive request from $z$ to its ancestor (the second case is proven symmetrically). We consider two cases, either $\Delta_i|f_z(t)| = \mathcal{E}+1$ or not. We first assume that it is the former case. 

\begin{itemize}
    \item $\bm{\Delta_i|f_z(t)| = |\Delta_i(f_z(t))| = \mathcal{E}+1}$: As noted earlier there are only 4 types of events that may cause $f_z(t)$ to change. By our assumption that the phase ends once the algorithm moves a positive request from $z$ to its ancestor we are guaranteed that the event of type (1) causes $f_z(t)$ to increase by 1. Since there are at most $\mathcal{E}+1$ events that may cause $f_z(t)$ to change and due to the fact that $|\Delta_i(f_z(t))| = \mathcal{E}+1$ we are guaranteed that $f_z(t)$ may only increase during the phase and that $\Delta_i(f_z(t)) = \mathcal{E}+1$. The following observation is always true.
    \begin{observation}
    $\Delta_i|f_z(t)| = |\Delta_i(f_z(t))| \Rightarrow sign(f_z(t_i)) = sign(\Delta_i(f_z(t)))$.
    \end{observation}
    Therefore, $f_z(t_i) \geq 0$ and throughout the phase we have that $sur_z^*(t) \geq sur_z(t)$. Thus, whenever $z^+(t)$ increases, we are guaranteed that $sur_z^*(t) \geq sur_z(t) > 0$. Note that always $o_z(t) \geq sur^*_z(t)$ and that whenever $z_z^+$ increases, it does so at a rate of $\alpha_z \cdot sur_z(t)$. Therefore, whenever $z^+$ increases we are guaranteed that,
    \begin{align*}
    o_z(t) \geq sur^*_z(t) \geq sur_z(t) >0.
    \end{align*}
    
    Recall that the phase ended once a positive request left the counter. Therefore, the overall increase in $z^+$ is exactly $y_z$. Thus, if we let $\mathcal{J}$ denote all times for which $z^+$ increased, then,
    \begin{align}
    \label{equation.bipartite.metric.yossis_lemma_2}
    y_z =
    \int_{t \in \mathcal{J}} \alpha_z sur_z(t)dt \leq
    \int_{t \in \mathcal{J}} \alpha_z o_z(t)dt \leq 
    \int_{t \in \mathcal{J}} \alpha_z o_z(t)dt,
    \end{align}
    where the last inequality follows from the positivity of $\alpha_z o_z(t)$.
    
    Since $\Delta_i(\phi_z) = 2y_z\Delta_i(|f_z(t)|)$, we are guaranteed that,
    \begin{align*}
    \int_{t \in I_i}z'(t)dt + \Delta_i(\phi_z) \leq 
    2y_z + \Delta_i(\phi_z) = 
    2y_z + 2y_z(\mathcal{E}+1) \leq 
    4\int_{t \in I_i}\alpha_z o_z(t)dt + 2y_z \mathcal{E},
    \end{align*}
    where the first inequality follows due to the fact that during a phase $z'(t)$ may increase by at most $2y_z$ ($y_z$ for $z^+$ and at most $y_z$ for $z^-$) and the second inequality follows from equation (\ref{equation.bipartite.metric.yossis_lemma_2}).

    \item $\bm{\Delta_i|f_z(t)| < \mathcal{E}+1}$: Due to the fact that $f_z(t)$ changes value exactly $\mathcal{E}+1$ times and therefore changes parity $\mathcal{E}+1$ times, we are guaranteed that in fact $\Delta_i|f_z(t)| < \mathcal{E}$. Therefore, since $\Delta_i(\phi_z) = 2y_z\Delta_i(|f_z(t)|)$,
    
    \begin{align*}
    \int_{t \in I_i}z'(t)dt + \Delta_i(\phi_z) \leq 
    2y_z + \Delta_i(\phi_z) \leq 
    2y_z + 2y_z(\mathcal{E}-1) \leq 
    4\int_{t \in I_i}\alpha_z o_z(t)dt + 2 y_z \mathcal{E},
    \end{align*}
    where the last inequality follows simply from $\alpha_z o_z(t)$'s positivity.
\end{itemize}

\noindent Now that we have proven the lemma for any phase other than the last, we consider the phase $I_{m-1} = [t_{m-1}, t_m = \infty)$. At the end of the phase we have that $sur^*_z(t_m) = sur_z(t_m)$ and therefore, $\Delta_{m-1}(\phi_z) \leq 0$. If $\mathcal{E} > 0$ then we are guaranteed that
\begin{align*}
\int_{t \in I_{m-1}}z'(t)dt + \Delta_{m-1}(\phi_z) \leq 2y_z \leq 2y_z \mathcal{E} \leq 4\int_{t \in I_{m-1}}\alpha_z o_z(t)dt + 2 y_z \mathcal{E}.
\end{align*}

If $\mathcal{E} = 0$ then $sur^*_z(t)$ remains constant throughout the phase and therefore $\phi_z(t) = 0$ throughout the phase. Therefore, $sur^*_z(t) = sur_z(t)$ whenever $z^+$ or $z^-$ increase. Furthermore, each of the requests counted by $o_z(t)$ incurs a momentary delay of at least $\alpha_z$. Since $|sur^*_z(t)| \leq o_z(t)$ and $z'(t) \leq \alpha_z|sur_z(t)|$ we are guaranteed that in fact $\int_{t \in I_{m-1}}z'(t)dt \leq \int_{t \in I_{m-1}}\alpha_z o_z(t)dt$. Therefore,
\begin{align*}
\int_{t \in I_{m-1}}z'(t)dt + \Delta_{m-1}(\phi_z) = \int_{t \in I_{m-1}}z'(t)dt \leq \int_{t \in I_{m-1}}\alpha_z o_z(t)dt \leq 4\int_{t \in I_{m-1}}\alpha_z o_z(t)dt + 2 y_z \mathcal{E},
\end{align*}
which completes the proof for the last phase as well.
\end{proof}

\end{document}